\documentclass[12pt,draftcls,onecolumn]{IEEEtran}
\usepackage{amsmath}
\usepackage{setspace}
\usepackage{color}
\usepackage{amssymb}
\usepackage{amsthm}
\usepackage{cite}
\usepackage{tikz}
\usepackage{caption}

\ifCLASSINFOpdf
\else
\fi
\begin{document}
%


\newtheorem{lemma}{Lemma}
\newtheorem{property}{Property}
\newtheorem{theorem}{Theorem}
\newtheorem{definition}{Definition}
\newtheorem{terminology}{Terminology}
\newtheorem{corollary}{Corollary}
\newtheorem{mytheorem}{Property}
\newtheorem{remark}{Remark}
\title{Uncertain Price Competition in a Duopoly with Heterogeneous Availability}

\author{Mohammad Hassan Lotfi,
        Saswati Sarkar,~\IEEEmembership{Senior Member,~IEEE}
\thanks{ M. H. Lotfi and Saswati Sarkar are with the Department of Electrical and Systems Engineering at University of Pennsylvania, Philadelphia,
PA, U.S.A. Their email addresses are lotfm@seas.upenn.edu and
swati@seas.upenn.edu respectively.}%
\thanks{Part of this work was presented in Allerton'12 \cite{Allerton}.}}
\maketitle

\vspace{-12mm}
\begin{abstract}
We study the price competition in a duopoly with an arbitrary number of buyers. Each seller can offer multiple units of a commodity depending on the availability of the commodity which is random and may be different for different sellers. Sellers seek to select a price that will be attractive to the buyers and also fetch adequate profits.  The selection will in general depend on the number of units available with the seller and also that of its competitor - the seller may only know the statistics of the latter. The setting captures a secondary spectrum access network, a non-neutral Internet, or a microgrid network in which unused spectrum bands, resources of ISPs, and  excess power units constitute the respective commodities of sale. We analyze this price competition as a game, and identify a set of necessary and sufficient properties for the Nash Equilibrium (NE). The properties reveal that sellers randomize their price using probability distributions whose support sets are mutually disjoint and in decreasing order of the number of availability. We prove the uniqueness of a symmetric NE in a symmetric market, and explicitly compute the price distribution in the symmetric NE.  

\end{abstract}


\IEEEpeerreviewmaketitle

\section{Introduction}


\subsection*{The Research Challenges and Goals}
\IEEEPARstart{W}{e} consider a market with two sellers, where each seller offers multiple commodities for sale. The commodities that are available for sale are randomly generated. In other words, sellers do not control the amount supplied or they may obtain the commodities from a residual supply. We investigate the price selection strategy for sellers in presence of uncertainty in competition using Game Theory~\cite{MWG}. Customers shop around for the lowest available prices. Therefore sellers seek to set prices that will ensure that their commodities are sold and also fetch adequate profit. In our model, a seller is not aware of the number of units available to her competitor before quoting her price. Thus, the competition that each seller faces is uncertain, and different sellers have different number of goods available (heterogeneous availability). Each seller selects the price per unit depending on the number of units she has available for sale, the statistics of the availability process for her competitor, and the demand. In general, each seller chooses her price randomly using different probability distributions for different availability levels.  Thus, the strategy of each player is a vector of probability distributions. For instance, if a seller can potentially offer up to three units of commodity, her vector of strategies would be $(\Phi_1(.),\Phi_2(.),\Phi_3(.))$, where $\Phi_i(.)$ is the price selection probability distribution when the seller offers $i$ units.



Due to uncertainty in competition, quoting a high price by a seller enhances the risk of not being able to sell the commodity offered by that seller. On the other hand, although selecting a low price increases the chance of winning the competition, it also decreases the profit earned by the seller. Therefore, pricing in presence of uncertainty in competition is a risk-reward tradeoff. It is not a priori clear that how offering multiple number of units affects the price selection by sellers. For instance, a seller with a large number of available units may be motivated to quote a low price, since in the event of winning the competition, a small amount of profit per unit would result in a large total profit.  On the other hand, a seller maybe enticed to select a high price when the availability is high to significantly increase her overall  profit, even at the risk of not being able to sell the available units.  We focus on investigating the impact of heterogeneous availability and uncertain competition on the aforementioned risk-reward tradeoff.

Uncertainty in competition is an integral feature of diverse sets of applications. In Section~\ref{application}, we outline the connection between the decision problem we considered and three different emerging  application domains: primary/secondary market, a non-neutral Internet, and microgrid networks. 

\vspace{-1mm}
\subsection*{Contributions}
We start by positioning our work in the context of the existing literature. We next model the price selection problem as a one-shot non-cooperative game (Section~\ref{system_model}). The sellers are allowed to have different probability distributions for different availability levels (asymmetric market). In Section~\ref{section:asymmetric_duopoly}, we identify key properties that every NE pricing strategy should satisfy when demand is greater than the maximum possible availability level. The properties reveal that the sellers randomize their price using probability distributions whose support sets are mutually disjoint and in decreasing order of the number of availability. In the context of the aforementioned risk-reward tradeoff,  sellers opt for low-risk pricing when they have high availability. In Section~\ref{section:duopolysufficiency}, we prove that any strategy profile that satisfies  the properties listed in Section~\ref{section:asymmetric_duopoly} constitutes an NE regardless of the relation between the demand and the number of available units. This sufficiency result naturally leads to an algorithm (Appendix~\ref{section:duopolyalgorithm_asym_framework}) for computing the strategies that satisfy the properties in Section~\ref{section:asymmetric_duopoly}. 

In Section~\ref{section:duopolyalgorithm}, we consider a symmetric market and prove that these properties are also necessary conditions for a NE regardless of the relation between the demand and the number of available units. We prove that the symmetric NE uniquely exists, and obtain an algorithm for explicitly computing it. Note that the uniqueness is specific to the symmetric market- our analysis in  Appendices~\ref{section:duopolyalgorithm_asym_example} and \ref{section:examples} reveals that an asymmetric market allows for multiple Nash equilibria. Results are generalized to the case of random demand in Section~\ref{section:random}.  
The asymptotic behavior of the symmetric NE (when $m\rightarrow \infty$) is investigated through numerical simulations  in Section~\ref{section:simulations}.



\vspace{-1mm}
\subsection*{Related Literature}\label{section:literatue}

Price competition among different entities has been extensively studied in \cite{commpricing,enerpricing,elecauction,swider2007bidding,Niyato1,Niyato2_MM,Zhou_TRUST,qual, fairelec, arnob}. In economics literature as also in the context of specific applications, uncertainty in competition has been investigated when the availability level is either zero or one \cite{J&R,kimmel,Gaurav_SPinCRN,Gaurav_Spectrum_White,Gaurav_microgrid}.  The strategy profile of each seller consists of only one probability distribution since sellers need to select a price only when they have one unit available for sale.  We, however, characterize the Nash equilibrium pricing strategies when sellers have arbitrary and potentially different number of available units for sale (not merely zero or one). In this case, different price selection strategies may be required for different number of available units. Thus, the pricing strategy profile of each seller is a collection of probability distributions, one for each availability value. Therefore both results and proofs are substantially different from previous works. 

Another genre of work allows sellers to control the amount of commodities they would generate for sale \cite{Couchman,Klemperer,Wang,federico2003bidding,holmberg2009supply,Anderson,underpricing,kastl}. In these works, sellers (e.g. power generators) bid a supply function \footnote{A supply function is a function that maps the price of the commodity under sale to the amount a producer will produce for sale.} to a central auctioneer. Given the demand and the bids submitted, the auctioneer solves an optimization problem to determine the number of units needed to be generated by the sellers and subsequently the price that should be paid to them. In \cite{Couchman,Klemperer, Wang}, the setting is  a uniform-price procurement auction in which the price is equal for different sellers, i.e. the clearing price. However, in \cite{federico2003bidding,holmberg2009supply,Anderson,underpricing,kastl}, authors investigate a pay-as-bid (discriminatory) procurement auction which is similar to our work in the sense that the price can be different for the bidder (sellers in our case). In these works, central entity accepts the offers submitted by the sellers and pays the accepted offers based on the bid submitted. 
For instance in \cite{Anderson}, authors  provide a characterization of mixed equilibria over increasing supply curves. In other words, in their characterization, the price per infinitesimal unit of the commodity  is increasing, i.e., the higher the number of units produced, the higher the price per unit. Note that in \cite{federico2003bidding,holmberg2009supply,Anderson} authors consider divisible goods, i.e. continuous amount of goods for sale. However, in \cite{underpricing} and \cite{kastl}, the number of units is effectively discrete. In this sense, they have a closer model to ours. Nonetheless, the main distinction of our work with this entire genre of work is that we consider scenarios where sellers do not control the amount of commodities they produce. This distinction in the setup, lead to major differences in the formulation, analyses, and results.




\vspace{-1mm}

\section{Market Model and Problem Formulation}\label{system_model}

\subsection{Market Model}

First, we define some preliminary notation. Then sellers' decision and information are described.

\subsubsection{Preliminary notation}

We consider a market with two sellers in which each seller owns multiple number of the same commodity and  quotes a price per unit. The total demand of the market is $d$ units. For simplicity, the demand is assumed to be deterministic. The generalization to random $d$ is straightforward, and is presented in Section~\ref{section:random}.

Buyers prefer the seller who quotes a lower price per
unit, and they are equally likely to buy a unit from sellers who select equal prices. Thus, if sellers have $a$, $b$ units to sell respectively and quote prices of $x$, $y$ per unit, where $x < y$, then they respectively sell $\min\{a,d\}$, $\min\{b, (d-a)^+\}$ units, where $z^+$ denotes $\max\{z, 0\}$. The cost of each transaction is $c$. Therefore, a seller earns a profit of $i(x-c)$ when she sells $i$ units with price $x$ per unit. Because of regulatory restrictions or because of valuations that buyers associate with purchase of each unit, the price selected by each seller should be bounded by some constant $v > c$, i.e. $x\leq v$. The availability of each seller is random:

\begin{terminology}
We denote $m_k$ as the maximum possible number of available units of seller $k$. Let $q_{kj}\in [0, 1]$ be the probability that seller $k$ has $j\in\{0,\dots,m_k\}$ units available, and $\vec{q}_k=(q_{k0},\dots,q_{km_k})$.
\end{terminology}

The availability of sellers may for example follow binomial distributions 
$\mathcal{B}(m_1,p_1)$ and $\mathcal{B}(m_2,p_2)$. Specifically, if $p_1=0.5$, $p_2=0.3$, $m_1=3$, and $m_2=2$, then $\vec{q}_1=(\frac{1}{8},\frac{3}{8},\frac{3}{8},\frac{1}{8})$ and  $\vec{q}_2=(\frac{49}{100},\frac{42}{100},\frac{9}{100})$.

We assume that sellers have zero unit available for sale with positive probability, i.e., $q_{k0}>0$ for $k\in\{1,2\}$, and the competition is uncertain, i.e., $q_{\bar{k}i}<1$ for $i\in\{0,1,\dots,m_k\}$  for at least one seller $k$. \footnote{Note that if this exists $i,j\in\{0,\dots,m_k\}$ such that $q_{1i}=1$ and $q_{2j}=1$, then both sellers know the exact number of available units with the other seller. Thus the competition is deterministic.} Note that if competition is deterministic for both sellers, then the problem is trivial.


\begin{terminology}
For each seller $k$, let $\bar{k}$ denote the other seller, i.e., if $k=1$ (respectively, $k=2$), then $\bar{k}=2$ (respectively $\bar{k}=1$).
\end{terminology}

\subsubsection{Sellers' decisions and information}

Sellers select their price based on the number of units they offer in the market.  Before choosing her price, a seller does not know the number of units of the commodity that her competitor has available
 for sale and the price per unit her competitor selects. She is however aware of the demand and the distributions for the above quantities. A seller may select her price randomly.

\begin{terminology}
Let $\Phi_{kj}(.)$ be the probability distribution that the seller $k\in \{1,2\}$ uses for selecting price per unit when she offers $j$ units. Let $\tilde{p}_{kj}$ and $\tilde{v}_{kj}$ be the infimum and the supremum of the \emph{support set}\footnote{The support set of a probability distribution is the smallest closed set such that its complement has probability zero
under the distribution function.} of $\Phi_{kj}(.)$. 
The strategy profile of seller $k$ is $\Theta_k(.)=(\Phi_{k1}(.),\dots,\Phi_{km_k}(.))$.
\end{terminology}

An example of probability distributions, support sets, and their infimums and supremums is presented in Figure~\ref{figure:dist}. In this figure, the infimums ($\tilde{p}_{kj}$'s) are illustrated explicitly, and $\tilde{v}_{kj}=\tilde{p}_{k,j-1}$ (For instance, $\tilde{v}_{13}=\tilde{p}_{12}$).  Note that, Figure~\ref{figure:dist} presents the distributions which are strictly increasing between the infimum and the supremum of their support sets. However, the probability distributions in general may consist of strictly increasing and flat parts. For example, a probability distribution that is strictly increasing over intervals $[a,b]$ and $[c,d]$, and flat over interval $[b,c]$.  Unlike the previous example, the support set of this probability distribution ($[a,d]\cup [c,d]$) is not connected.


\vspace{-11pt}
\subsection{Problem Formulation}
Note that in general, the number of units a seller sells and her profit can be random.

\begin{terminology}
Let  $u_k (\Theta_k(.), \Theta_{\bar{k}}(.))$ denotes  the expected profit of seller $k$ when she adopts strategy profile $\Theta_k(.)$ and her competitor adopts $\Theta_{\bar{k}}(.)$. 
\end{terminology}

 In a Bayesian game (where players are modeled as risk-neutral), rational players are seeking to maximize their expected payoff, given their beliefs about the other players (

\begin{definition}
 A \emph{Nash equilibrium} (NE) \footnote{Clearly, our game is a Bayesian game with the number of available units for sale being the type of a player. For the sake of notational convenience, we use Nash equilibrium as an alternative for Bayesian Nash equilibrium.} is a strategy profile such that no seller can improve her expected profit by unilaterally deviating from her strategy. Therefore, $(\Theta_1^\star(.),\Theta_2^\star(.))$ is a NE if for each seller $k$:
\begin{equation}
u_k(\Theta^*_{k}(\cdot),\Theta^*_{\bar{k}}(\cdot))\geq u_k(\tilde{\Theta}_{k}(\cdot),{\Theta}^*_{\bar{k}}(\cdot)),\quad \forall \quad \tilde{\Theta}_k(\cdot).
\nonumber
\end{equation}
\end{definition}

\begin{terminology}
With slight abuse of notation,
we denote $u_{kl}(x)$ as the expected profit  that seller $k$ earns, and
  $B_{kl}(x)$ as the expected number of units that seller $k$ sells, when she offers $l$ units for sale with price $x$ per unit, respectively (the dependence on the competitor's strategy is implicit in this simplified notation).
\end{terminology}
\vspace{-3mm}
\begin{equation}\mbox{Clearly, } u_{kl}(x) = B_{kl}(x) (x-c). \label{equ:utility} \end{equation}


Note that $\frac{u_{kl}}{l}$ is the expected utility per unit of availability. Thus, $A_{k,l,j}(x) = \frac{1}{l}u_{kl}(x)-\frac{1}{j}u_{kj}(x)$ is the difference  between the utility per availability for availability levels $l$ and $j$. We will see that $A_{k,l,j}(x)$ plays an important role throughout in the proofs, which motivates the following terminology:

\begin{terminology}\label{terminology:Aklj}
Let $A_{k,l,j}(x) = \frac{1}{l}u_{kl}(x)-\frac{1}{j}u_{kj}(x)=(x-c)B_{k,l,j}(x) $, where $B_{k,l,j}(x) =\frac{1}{l} B_{kl}(x) - \frac{1}{j} B_{kj}(x)$.
\end{terminology}

\begin{terminology}
 Let $e_k=(d-m_{\bar{k}})^+$.
\end{terminology}

Note that for all $x \leq v$,
\begin{equation}
B_{kl}(x)=l\quad \quad l=1,\dots,e_k \label{equ:sellall}
\end{equation}
as $k$ will sell all she offers in this case given that the total offering is less than the demand. We would later obtain the expression for $B_{kl}(x)$ under  the NE strategy profiles when $ l >  e_k.$


\begin{definition}
A price $x$ is said to be a \emph{best-response price} for seller $k$ when she offers $j$ units if $u_{kj}(x) \geq u_{kj}(a)$ for all $a \in [0, v].$
\end{definition}

Note that a NE-strategy profile selects with positive probability only amongst the best-response prices. Thus, all the elements of support sets are best responses except potentially those on the boundaries (elements of boundaries may not be best responses) if there is a discontinuity in the utility at those points.

We seek to determine the Nash equilibrium strategy profile of sellers. If $m_1+m_2\leq d$, since there is no competition between sellers, both sellers offer their units with the monopoly price, $v$
at the NE. We therefore assume that $m_1+m_2>d$.



\section{Properties of a NE when $d > \max\{m_1,m_2\}$}\label{characduopoly_asym} \label{section:asymmetric_duopoly}

Note that from Corollary 5.2 in \cite{existence}, a mixed strategy NE exists for our model. In this section, we investigate the necessary conditions for a strategy to be an NE when $d > \max\{m_1,m_2\}$ (Theorem \ref{theorem:necessary}). We will explicitly point out whenever we use the assumption that  $d > \max\{m_1,m_2\}$. 

\begin{theorem}\label{theorem:necessary}
A NE must satisfy the following properties when $d> \max\{m_1,m_2\}$,
\begin{enumerate}
\item \label{prop:v} For each $k$, there exists a threshold  such that seller $k$ offers price $v$ with probability one if she has the availability level less than or equal to this threshold. This threshold, denoted as $l_k$ henceforth, is such that:
\begin{enumerate}
\item \label{prop:v1} $l_k\in\{e_k,\dots,m_k-1\}$
\item \label{prop:v2}$l_1+l_2=d-1$ or $l_1+l_2=d$
\end{enumerate}
\item \label{prop:2} When seller $k$ has $l_k+1$ \footnote{The same $l_k$ as the one in part 1.} units, she uses distribution $\Phi_{k,l_k+1}(.)$
\begin{enumerate}
\item \label{prop:2intermsupportset} whose support set is $[\tilde{p}_{k,l_k+1}, v]$,
\item  \label{prop:2intermcont} which is continuous throughout except possibly at $v$, and
\item \label{prop:2jumpv} has a jump at $v$ for at most one value of  $k\in\{1, 2\}$, and size
of such a jump is less than $1$
\end{enumerate}
\item \label{prop:3} When the availability level is $i\in\{l_k+2,\dots,m_k\}$\footnote{The same $l_k$ as the one in part 1}, seller $k$ uses distribution $\Phi_{ki}(.)$
\begin{enumerate}
\item \label{prop:2supportset} whose support set is $[\tilde{p}_{k,i}, \tilde{p}_{k, i-1}]$,
\item  \label{prop:2cont} which is continuous throughout
\item \label{prop:2last} $\tilde{p}_{1, m_k} = \tilde{p}_{2, m_k}$
\end{enumerate}
\item \label{prop:equal} The utility of seller $k$ when she offers $i$ units is equal for all  prices in the support set of $\Phi_{ki}(.)$, except possibly at price $v$ (if $v$ belongs to her support set).
\end{enumerate}
\end{theorem}

In Appendix~\ref{section:duopolyalgorithm_asym}, we will present an algorithm to explicitly compute the NE strategies satisfying properties in Theorem~\ref{theorem:necessary}.  Using this algorithm, in Figure~\ref{figure:dist}, we plot an NE probability distribution of price when the vector of availability  distributions are $\vec{q_1}=[0.3,0.2,0.2,0.3]$ and $\vec{q_2}=[0.4,0.2,0.2,0.2]$, the demand , i.e. $d$, is $3$, $v=10$, and $c=6$. Note that in this case $l_1=l_2=1$, and $l_1+l_2=d-1$ (part~\ref{prop:v} at Theorem~\ref{theorem:necessary}). This means that both sellers offer price $v$ with probability one if they have one unit of commodity available. When sellers have availability $l_1+1=2$ and $l_2+1=2$ units available for sale, they use probability distributions $\Phi_{12}(.)$ and $\Phi_{22}(.)$, respectively, whose support sets are $[\tilde{p}_{12},v]$ and $[\tilde{p}_{22},v]$, respectively (part \ref{prop:2intermsupportset} of the Theorem). In addition, these distributions are continuous throughout except possibly at $v$ (part~\ref{prop:2intermcont}). Furthermore, only the probability distribution $\Phi_{22}(.)$ has a jump at price $v$ and the size of this jump is less than one (part~\ref{prop:2jumpv} of Theorem~\ref{theorem:necessary}). When sellers have availability level $l_1+2=l_2+2=3$, they use probability distributions $\Phi_{13}(.)$ and $\Phi_{23}(.)$, respectively, whose support sets are $[\tilde{p}_{13},\tilde{p}_{12}]$ and $[\tilde{p}_{23},\tilde{p}_{22}]$, respectively (part \ref{prop:2supportset} of Theorem~\ref{theorem:necessary}). In addition, these probability distributions are continuous throughout (part~\ref{prop:2cont}). Note that $\tilde{p}_{13}=\tilde{p}_{23}=\tilde{p}$ (part~\ref{prop:2last} of the Theorem). More numerical examples are presented in Appendix~\ref{section:duopolyalgorithm_asym_example}. 


\begin{figure}[t]
\begin{center}
\includegraphics[scale=0.43]{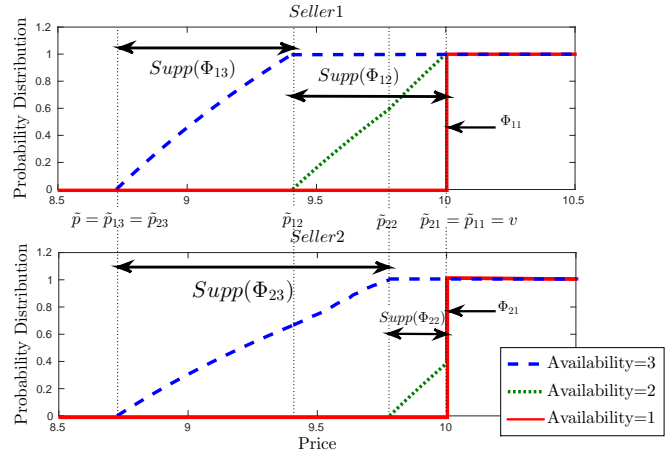}
\end{center}
\caption{An example of an NE pricing strategy, when $d=3$, $Supp=Support\ Set$. Note that $\Phi_{11}$ and $\Phi_{21}$ have a jump of magnitude one, and $\Phi_{22}$ has a jump of size $0.6$ at $v$.}
\label{figure:dist}
\end{figure}

We prove Theorem~\ref{theorem:necessary} using the following results which we first state and prove later.


\begin{enumerate}
\item \label{sum1}  The probability distribution of price, $\Phi_{ki}(x)$ for $i\in\{1,\dots,m_k\}$, is continuous for $x<v$ (Section~\ref{continuity}, Property~\ref{property:continuous_as}).
\item \label{sum2} The lower bound of prices are equal for both sellers (Section~\ref{lowerbound}, Property~\ref{property:equal_lower_bound}).
\item \label{sum3} There is no gap between support sets (Section~\ref{gap}, Property~\ref{property:contiguous_assym}).
\item \label{sum4} Support sets are disjoint barring common boundary points, and are in decreasing order of the number of available units for sale (Section~\ref{supportset}, Property \ref{property:supportsets}).
\item \label{sum5} The structure of NE at price $v$: A seller selects $v$ with probability one, if and only if the number of available units with her is less than or equal to a threshold $l_k\in\{0,1,\dots,m_k-1\}$, where $l_1+l_2=d$ or $l_1+l_2=d-1$ (Section~\ref{subsection:difference}, Property~\ref{property:difference}).
\end{enumerate}

Note that in Figure~\ref{figure:dist}, the distributions are continuous and the lower bound of prices are equal. In addition, every element of the set $[\tilde{p},v]$ belongs to a support set, i.e. there is no gap between support sets. The support sets of seller one when she offers $3$, $2$, and $1$ unit are $[\tilde{p},\tilde{p}_{12}]$, $[\tilde{p}_{12},v]$, and $\{v\}$, respectively. This illustrates the result $4$. The result $5$ is the same as  part~\ref{prop:v} in Theorem~\ref{theorem:necessary}, and is previously connected to Figure~\ref{figure:dist}.

Henceforth in this section, we focus on proving the necessary results and properties  needed to prove Theorem~\ref{theorem:necessary}. 





\subsection{Results that we use throughout}

\begin{property}\label{property:c_asy}
For each $i$ and $k$, $\Phi_{ki}(c) = 0$.
\end{property}

This result follows directly since prices less than cost $c$ are not chosen by sellers. Property~\ref{property:c_asy} therefore rules out jumps at prices $x\leq c$.

\begin{proof}
Note that for each $i$, $u_{ki}(x) \leq 0$ for $x \leq c$. But, since $B_{ki}(x)\geq iq_{\bar{k}0} > 0$ for all $x  \in [0, v]$, $u_{ki}(x) > 0$  for all $x \in (c, v]$. Thus, no price in $[0, c]$ is a best response for a seller.
\end{proof}

Lemma~\ref{lemma:jump_asy}, which we use throughout the paper,  rules out jumps at prices higher than $c$. 

\begin{lemma}\label{lemma:jump_asy}
Let the strategy profile of player $k$ be $\Theta_k(.)=(\Phi_{k1}(.),\dots,\Phi_{km_k}(.))$, and $\Phi_{ki}(.)$ have a jump at $x > c$. Then for $l$ such that $l + i > d$, $u_{\bar{k}l}(x -\epsilon') > u_{\bar{k}l}(a)$, $\forall a \in[x, \min\{x + \epsilon, v\}]$, and for all sufficiently small but positive $\epsilon$ and $\epsilon'$. 
\end{lemma}

We provide the intuition behind the result and defer the proof to Appendix~\ref{appendix:lemma1}.
Note that offering a lower price increases the expected number of units sold by a seller, but decreases the revenue per unit sold. Suppose that a seller $k$ offers $i$ units with price $x$ with a positive probability. Let her competitor $\bar{k}$ have $l$ units available where $l+i>d$; $\bar{k}$ can sell  a strictly larger  number of units in an expected sense  by choosing a price
 in the left neighborhood of $x$ (eg, $x-\epsilon$)  rather than $x$ or in its right neighborhood. In addition the difference is bounded away from zero even as the size of the left neighborhood approaches zero. On the other hand, the difference in the revenue per unit  approaches zero as the size of the left neighborhood approaches zero.  Therefore, prices in the left neighborhood of $x$  constitute better responses for  the seller than $x$ or those in its right neighborhood.

The following property fully characterizes the NE when seller $k$ offers $i\in\{1,\dots,e_k\}$ units.

\begin{property}\label{property:e_asy}
$\Phi_{ki}(x)$ selects $v$ with probability $1$ and any other prices with probability $0$ when
$i = 1,\dots,e_k$ for each $k$.
\end{property}

The proof relies on the fact that if a seller offers less than or equal to $e_k$ units of commodity, she can sell all units regardless of the price she quotes. Therefore $v$ strictly dominates all other prices.

\begin{proof}
This statement holds by vacuity if $\max\{m_1,m_2\}\geq d$. Now consider $d > \max\{m_1,m_2\}$. If the seller $k$ offers $i\leq e_k$ units, the total offerings from both sellers are at most $d$, since the other seller offers at most $m_{\bar{k}}$ units. Thus, the seller $k$ can sell everything it offers with any price $x$ in interval $[0,v]$. Therefore for all $x \in [0, v)$,  $u_{ki}(x) =i (x-c) < i(v-c) = u_{ki}(v)$. Thus, no price in $[0, v)$ is a best response. The result follows.
\end{proof}

\subsection{Continuity of Price Distribution for Price $x<v$}
\label{continuity}
Utilizing Lemma~\ref{lemma:jump_asy}, we can prove that the distribution of price is continuous for prices less than $v$,
\begin{property}\label{property:continuous_as}
$\Phi_{ki}(x)$ is continuous for $x<v$.
\end{property}

Note that in Fig~\ref{figure:dist}, there is no jump in the distributions for prices less than $v$.  


\begin{proof}
If $i\leq e_k$, the property follows from Property \ref{property:e_asy}. Now let $i>e_k$. If $x\leq c$, the property follows from Property \ref{property:c_asy}. Now consider $x \in (c, v)$. We use contradiction argument. Suppose $\Phi_{ki}(.)$ has a jump at price $x< v$. Since $i>e_k$, there exists $l \leq m_{\bar{k}}$ such that $l + i > d$. Using lemma \ref{lemma:jump_asy}, we can say that if $\Phi_{ki}(.)$ has a jump at $x$, for each $l$ such that $l + i > d$, $u_{\bar{k}l}(x -\epsilon')> u_{\bar{k}l}(a)$, where $a\in [x, \min\{x + \epsilon, v\}]$, and for all sufficiently small but positive $\epsilon$ and $\epsilon'$. Therefore no price in this interval is a best response for the seller $\bar{k}$ when she offers $l$ units. Therefore $\Phi_{\bar{k}l}(x + \epsilon) = \Phi_{\bar{k}l}(x)$ for all sufficiently small but positive $\epsilon$ and all $l$ such that $l > d-i$, i.e. the other seller does not choose any price in $[x, x + \epsilon)$ whenever she offers $l$ units. Knowing this we can say that $B_{ki}(a) = B_{ki}(x)$ for all $a \in [x, x + \epsilon)$ for some $\epsilon>0$ such that $x + \epsilon \leq v$. Therefore,

\begin{equation}
u_{ki}(x)=(x-c)B_{ki}(x)<(x+\frac{\epsilon}{2}-c)B_{ki}(x+\frac{\epsilon}{2})=u_{ki}(x+\frac{\epsilon}{2})
\end{equation}

\normalsize
Thus, $x$ is not a best response for a seller who offers $i$ units. Hence $x$ is chosen with probability zero, which rules out a jump at $x$ for $\Phi_{ki}(.)$. The property follows.
\end{proof}

Based on this property, the distribution of price is continuous for $x<v$. We will later show that the price distribution has a jump at $v$ for some availabilities.

Based on the above continuity result, the expression for the expected number of units sold for all $x \in [0, v)$ and $l=e_k+1,\dots,m_k$ is,

\small
\begin{equation}\label{equation:utility}
\begin{aligned}
B_{kl}(x)=l\sum_{i=0}^{d-l}&{q_{\bar{k}i}}+l\sum_{i=d-l+1}^{m_{\bar{k}}}{\big{(}1-\Phi_{\bar{k}i}(x)\big{)}q_{\bar{k}i}}\\
&+\sum_{i=d-l+1}^{m_{\bar{k}}}{\Phi_{\bar{k}i}(x)q_{\bar{k}i}(d-i)}
\end{aligned}
\end{equation}

\normalsize
Note that we assumed $d\geq \max\{m_1,m_2\}$ in \eqref{equation:utility}. The first term in the left hand side  corresponds to the situation in which the other seller offers
at most $d-l$ units. In this case, seller $k$ will sell all $l$ units she offered in the market. The second and the third terms are corresponding to the situation in which the other seller offers more than $d-l$ units with a price higher than and less than $x$, respectively. If the other seller offers with price higher than $x$, seller $k$ is able to sell the entire $l$ units. On the other hand, if $\bar{k}$ offers with a price less than $x$, $k$ will sell $d-l$ units of commodity.

We can now obtain an expression for $u_{kl}(x)$ for $x < v$ from \eqref{equ:utility}, \eqref{equ:sellall}, and \eqref{equation:utility}.

\subsection{Sellers Have Equal Lower Bound of Prices}
\label{lowerbound}

Note that the example NE distributions presented in Figure~\ref{figure:dist} have equal lower bounds ($\tilde{p}=\tilde{p}_{13}=\tilde{p}_{23}$). We now prove that all NE distributions must satisfy this property:
 
\begin{property}\label{property:equal_lower_bound}
The minimum of lower end points of support sets are equal for both sellers. Mathematically,
$$
\tilde{p}_1=\tilde{p}_2
$$
where, $\tilde{p}_k=\min\{\tilde{p}_{ki}: i=1,\dots,m_k\}$. Furthermore, $\tilde{p}_1=\tilde{p}_2<v$ if $d<m_1+m_2$.
\end{property}

If the lower bound of prices for seller $k$, i.e. $\tilde{p}_k$, is lower than that for the other seller, $\tilde{p}_{\bar{k}}$, then  $k$ sells equal number of units in an expected sense by choosing $\tilde{p}_k$ as any other price in $(\tilde{p}_k, \tilde{p}_{\bar{k}})$. Using continuity of distributions for prices less than $v$, we can say that $\tilde{p}_{\bar{k}}$ is a better response than $\tilde{p}_k$ for $k$, which is a contradiction. The formal proof follows:

\begin{proof}
Suppose not. Without loss of generality suppose $\tilde{p}_1<\tilde{p}_2\leq v$. Therefore there exists $j$ such that $\tilde{p}_1$ belongs to the support set of $\Phi_{1j}(.)$. Since player 2 does not offer with any price in the interval $[\tilde{p}_1,\tilde{p}_2)$, $B_{1j}(\tilde{p}_1)=B_{1j}(\tilde{p}_2^-)$ \footnote{$f(x^-)=\lim_{y\uparrow x} f(y)$}. Thus $u_{1j}(\tilde{p}_1)<u_{1j}(\tilde{p}_2^-)$ which contradicts the assumption that $\tilde{p}_{1}$ is a best response for the first player when she offers $i$ units of commodity. Therefore, the first part of the property follows.

Suppose $\tilde{p}_1=\tilde{p}_2=v$. Thus, both sellers choose the price $v$ with probability $1$ regardless of the number of units they have available. Consider
seller $k$. Let $l=m_{\bar{k}}$. Since $m_1+m_2>d$, Lemma \ref{lemma:jump_asy} implies that $u_{\bar{k}m_{\bar{k}}}(v-\epsilon)>u_{\bar{k}m_{\bar{k}}}(v)$. This contradicts the assumption that $v$ is the best response for seller $k$. The result follows.
\end{proof}

\begin{terminology}
Let $\tilde{p}$ denote the minimum of lower end points of prices in the NE, i.e. $\tilde{p}_1=\tilde{p}_2=\tilde{p}$.
\end{terminology}

\subsection{The union of support sets cover $[\tilde{p}, v]$}
 \label{gap}

We show that there does not exist an interval of prices in $[\tilde{p}, v]$ which is eschewed with probability $1$ by both sellers. If such an interval existed, the cumulative distribution functions of both sellers would be flat in it, which we rule out below. Note that in Fig~\ref{figure:dist}, the NE distributions are strictly increasing throughout their support sets, and there is no flat region.

\begin{property}\label{property:contiguous_assym}
There does not exist $a$, $b$ such that $\tilde{p}\leq a<b\leq v$ and $\Phi_{ki}(b)=\Phi_{ki}(a)$ for all $i\in \{e_k+1,\dots, m_k\}$ and $k=1,2$.
\end{property}

If such $a$ and $b$ exist for seller $k$, this means that regardless of the number of available units, $k$ does not select any price in the interval $(y,z)$ where $y\leq a$, $z\geq b$, and $y$ is a best response 
when $k$ has an availability level $l$. This implies that for the competitor, $\bar{k}$, the expected utility is strictly  increasing in interval $[y,b]$. Thus $\bar{k}$ does not select any price in the interval $[y,b)$. This again implies that for seller $k$, when she offers $l$ units, price $b$ yields a strictly higher payoff than $y$, which is in contradiction with $y$ being a best response for $k$ when offering $l$ units. The formal proof is as follows:   

\begin{proof}
Let there be $a$, $b$, and $k$ such that $\tilde{p}\leq a<b\leq v$ and $\Phi_{ki}(b) = \Phi_{ki}(a)$ for all $i$. Thus for $\zeta$ such that $a<b-\zeta<b\leq v$, $\Phi_{ki}(b-\zeta) = \Phi_{ki}(a)$. Consider $y$ such that,
\begin{equation}
y=\sup\{x|x<a, x \in \text{support set of $\Phi_{kl}(.)$ for an $l$}\}
\nonumber
\end{equation}
Since support sets are closed, $y$ belongs in the support set of $\Phi_{kl}(\cdot)$ for some $l$. Thus, $y$ is a best response when the availability of player $k$ is $l$ (using Property~\ref{property:continuous_as} and $y < v$).

In addition, note that $\Phi_{ki}(y) = \Phi_{ki}(b-\zeta)$ for all $i$. Since $a < b-\zeta < v$, from Property \ref{property:continuous_as} and equation (\ref{equation:utility}), the expected number of units sold for the second seller remains constant for prices in $[y, b-\zeta]$, regardless of the number of units she offers,  i.e. $B_{\bar{k},.}(y) = B_{\bar{k},.}(b-\zeta)$. Thus, $u_{\bar{k},.}(b-\zeta)>u_{\bar{k},.}(y)$, and player $\bar{k}$ does not offer any price in the interval $[y,b-\zeta)$. Therefore $\Phi_{\bar{k},.}(y)=\Phi_{\bar{k},.}(b-\zeta)$. Since $a < b-\zeta < v$, from Property \ref{property:continuous_as} and equation (\ref{equation:utility}), $B_{kl}(y) = B_{kl}(b-\zeta)$. Thus, $u_{kl}(b-\zeta)>u_{kl}(y)$. This is in contradiction with $y$ being a best response  when the availability of player $k$ is $l$. Therefore, there does not exist $a$, $b$ such that $\tilde{p}\leq a<b\leq v$ and $\Phi_{ki}(b)=\Phi_{ki}(a)$ for all $i\in \{1,\dots, m_k\}$ and $k=1,2$. Also, note that for $i\in\{1,\dots,e_k\}$, $\Phi_{ki}(b)=\Phi_{ki}(a)$ for $\tilde{p}\leq a<b\leq v$, since support sets for these distributions only contain $v$. The result follows.
\end{proof}

\label{supportset}

\noindent \textbf{Remark:} In all the previous results, we considered $d\geq \max\{m_1,m_2\}$. In the next section, we need to consider that $d>\max\{m_1,m_2\}$.

 \subsection{Support Sets Are Mutually Disjoint and in Decreasing Order of the Number of Availabilities}\label{subsection:dgeqm} 
We start with proving a result, Lemma \ref{lemma:decreasing_random_asym}, on $A_{k,l,j}(x)$ (defined in Section~\ref{system_model}, Terminology~\ref{terminology:Aklj}). Note that we use Lemma~\ref{lemma:decreasing_random_asym} in subsequent sections as well. We next prove Property~\ref{property:supportsets} using this result, which leads to the main results of this section: Corollaries~\ref{corollary:disjoint_asym} and \ref{corollary:structure}. 

 First, using \eqref{equ:sellall} and \eqref{equation:utility}.

\vspace{-10pt}
\begin{equation}\label{equation:diff_decreasing}
\begin{aligned}
B_{k,l,j}(x)=&-\frac{1}{l}\sum_{i=d-l+1}^{d-j}{\Phi_{\bar{k}i}(x)q_{\bar{k}i}(i-d+l)}\\
&+\sum_{i=d-j+1}^{m_{\bar{k}}}{\Phi_{\bar{k} i}(x)q_{\bar{k} i}(d-i)(\frac{1}{l}-\frac{1}{j})}
\end{aligned}
\end{equation}

Thus, $B_{k, l, j}(\cdot)$ is non increasing and non positive with respect to the price $x$ when $l> j$. Therefore if $l > j$ then $A_{k,l,j}(x)$ is non increasing and non positive with respect to $x$. Based on the following lemma, $A_{k,l,j}(x)$ is (strictly) decreasing for $v>x\geq\tilde{p}$ and $l>j$ if $d> \max\{m_1,m_2\}$.

\begin{lemma}\label{lemma:decreasing_random_asym}
For each seller $k\in\{1,2\}$ and every $l$ and $j$, $j<l\leq m_k$, $A_{k,l,j}(x)$ is (strictly) decreasing for $\tilde{p}\leq x<v$ when $d
> \max\{m_1,m_2\}$.
\end{lemma}
Lemma~\ref{lemma:decreasing_random_asym} implies that the high-availability agent sacrifices more expected payoff  than the low-availability agent by increasing her price. 

Since $A_{k,l,j}(.)=(x-c)B_{k,l,j}(x)$, knowing that $B_{k,l,j}(x)$ is non-increasing, lemma follows if we prove that  $B_{k, l, j}(\cdot)$ is negative. We will prove that $\Phi_{km_k}(x)$, which is included in the summation of $B_{k, l, j}(\cdot)$, is positive 
for  $x> \tilde{p}$ and $k\in\{1,2\}$. In addition, the coefficient of $\Phi_{km_k}(x)$ is negative since $d> \max\{m_1,m_2\}$. Thus, the result follows.  

\begin{proof}
It is enough to prove that $B_{k,l,j}(x)$ is non-increasing for $x\geq \tilde{p}$ and negative for $x>\tilde{p}$. This yields that  $A_{k,l,j}(x)=(x-c)B_{k,l,j}(x)$ is strictly decreasing with respect to $x$. 

Note that in (\ref{equation:diff_decreasing}), $\Phi_{\overline{k}j}(.)$'s are non-negative and non-increasing since they are probability distributions. In addition, they have negative weights: $-(i-d-l)\leq -1<0$, $\frac{1}{l}-\frac{1}{j}<0$, and since $d>\max\{m_1,m_2\}$, $d-i\geq d-m_{\bar{k}}>0$. Thus $B_{k,l,j}(x)$ is non increasing and non positive with respect to the price $x$ when $l\geq j$. To prove that $B_{k,l,j}(x)$ is negative for $x>\tilde{p}$, since the distributions in (\ref{equation:diff_decreasing}) have (strictly) negative weights , it is enough to prove that at least one of the $\Phi_{\overline{k}j}(.)$'s  is included in the summation of $B_{k,l,j}(.)$ is positive, i.e. not all of them are zero.  We will prove that $\Phi_{km_k}(x)>0$ for  $x> \tilde{p}$ and $k\in\{1,2\}$. 

Suppose not and there exists $x>\tilde{p}$ such that $x\leq \tilde{p}_{km_k}$. By Property \ref{property:contiguous_assym}, there exists an $\epsilon>0$ and an availability level $j\neq \{1,\dots,e_k,m_k\}$ such that $[\tilde{p}_{km_k}-\epsilon,\tilde{p}_{km_k}]$ belongs to the support set of $\Phi_{kj}(.)$ and $\tilde{p}_{kj}<\tilde{p}_{km_k}$.  Thus $u_{kj}(\tilde{p}_{km_k})=u_{kj}(\tilde{p}_{km_k}-\epsilon)$. In addition, $B_{k,m_k,j}(x)$ is the weighted summation of $\Phi_{\bar{k}i}(.)$ for $i\in\{e_{\bar{k}}+1,\dots,m_{\bar{k}}\}$. Property \ref{property:contiguous_assym} implies that $\tilde{p}_{kj}$ belongs to at least one of the support sets of $\Phi_{\bar{k}i}(.)$ for $i\in\{e_{\bar{k}}+1,\dots,m_{\bar{k}}\}$. The distribution $\Phi_{\bar{k}i}(.)$ is included in the summation of $B_{k,m_k,j}(x)$, and its coefficient is negative.
 Thus, $A_{k,m_k,j}(x)$ is strictly decreasing with respect to $x$ for $x>\tilde{p}_{kj}$. Thus $A_{k,m_k,j}(\tilde{p}_{km_k}-\epsilon)>A_{k,m_k,j}(\tilde{p}_{km_k})$. Using $u_{kj}(\tilde{p}_{km_k})=u_{kj}(\tilde{p}_{km_k}-\epsilon)$, we can conclude that $u_{km_k}(\tilde{p}_{km_k})=u_{km_k.max}< u_{km_k}(\tilde{p}_{km_k}-\epsilon)$. This contradicts with $\tilde{p}_{k,m_k}$ belonging to the support set of $\Phi_{km_k}(.)$. The result follows.
\end{proof}

Note that in the previous lemma, we used $d> \max\{m_1,m_2\}$ to prove that $A_{k,l,j}(x)$ is decreasing for $\tilde{p}\leq x <v$. The following properties characterize the NE for price less than $v$.

\begin{property}\label{property:supportsets}
For $k\in\{1,2\}$, the support set of $\Phi_{kl}(.)$ is a subset of $[\tilde{p},\tilde{p}_{kj}]\cup [v]$ for all integers $j \in [1,l)$. 
\end{property}
For example, in Figure~\ref{figure:dist}, the support set for seller 1 and availability 3  is $[\tilde{p},\tilde{p}_{12}]$, which is a subset of the mentioned set.

\begin{proof}
First note that for $j\in\{1,\dots,e_k\}$ property follows, since $\tilde{p}_{kj}=v$ by Property \ref{property:e_asy}. Now consider $j>e_k$. Consider support sets of $\Phi_{kj}(\cdot)$, $\Phi_{kl}(\cdot)$, and $j<l$. We will show that $u_{kl}(a)< u_{kl}(\tilde{p}_{kj})$ for all $a\in (\tilde{p}_{kj}, v)$. Thus, no $a\in (\tilde{p}_{kj},v)$ is a best response for the seller $k$ with availability of $l$ units. Therefore, the support set of $\Phi_{kl}(\cdot)$ is a subset of $[c,\tilde{p}_{kj}]\cup [v]$.

We now complete the proof, by showing that $u_{kl}(a) < u_{kl}(\tilde{p}_{kj})$ for all $a \in(\tilde{p}_{kj}, v)$:
\begin{equation}
\begin{aligned}
\frac{1}{l}u_{kl}(a)-\frac{1}{j}u_{kj}(a)&=A_{k,l,j}(a)
\end{aligned}
\nonumber
\end{equation}
Since $l > j$ and $\tilde{p}\leq \tilde{p}_{kj}<a<v$, by Lemma \ref{lemma:decreasing_random_asym}, $A_{k,l,j}(a)$ is decreasing function of $a$ for $a \in [\tilde{p}_{kj}, v)$. Thus, $A_{k,j}(a) < A_{k,j}(\tilde{p}_{kj})$ for $a\in(\tilde{p}_{kj}, v)$. On the other hand $u_{kj}(a)\leq u_{kj}(\tilde{p}_{kj})$ for all $a>\tilde{p}_{kj}$ , since $\tilde{p}_{kj}$ is a best response of a seller with availability $j$, therefore $u_{kl}(\tilde{p}_{kj})>u_{kl}(a)$.
\end{proof}

Note that, in this stage, since $\Phi_{kl}(.)$ can have a jump at $v$, we cannot rule out $v$ as a member of the support set of $\Phi_{kl}(.)$.

\begin{corollary}\label{corollary:disjoint_asym}
The support sets of $\Phi_{kl}(.)$ and $\Phi_{kj}(.)$ overlap at most at one point in $[\tilde{p}, v)$.
\end{corollary}

For instance, note that in Figure~\ref{figure:dist}, the support sets of
$\Phi_{13}$ and $\Phi_{12}$ overlap only at $\tilde{p}_{12}$, the support sets of $\Phi_{12}$ and $\Phi_{11}$ overlap only at $v$, and there is no overlap between support sets of $\Phi_{13}$ and $\Phi_{11}$.

\begin{proof}
Suppose two points $x_1$ and $x_2$, where $x_1<x_2<v$, and both points  belong to the intersection of the support sets of $\Phi_{kj}(\cdot)$ and $\Phi_{kl}(\cdot)$.  Without loss of generality, consider $j<l$. The price $x_2>\tilde{p}_j$ belongs to the support set of $\Phi_{kl}(.)$, which is a contradiction with Property \ref{property:supportsets}.
\end{proof}

\begin{corollary}\label{corollary:structure}
For prices less than $v$ support sets are contiguous (Property~\ref{property:contiguous_assym}), disjoint (except possibly at one point) (Corollary~\ref{corollary:disjoint_asym}), and in decreasing order of the number of available units for sale (Property~\ref{property:supportsets}). Thus, there exists an increasing sequence $a_{km_k}, a_{k,m_k-1}, \ldots$ of positive real numbers in  $(c, v]$ such that the seller $k$ will randomize her price in the interval $[a_{ki},a_{k,i+1}]$ and possibly $\{v\}$ when she has $i$ units of commodity available for sale. 
\end{corollary}

For instance, note that in Figure~\ref{figure:dist}, the support sets of seller one are in decreasing order of the number of available units for sale, and the aforementioned increasing sequence is $\tilde{p},\tilde{p}_{12}$, and $v$.


\subsection{The Structure of Nash Equilibrium at Price $v$}
\label{subsection:difference}

We will investigate the possibility of having a jump at $v$. First, we prove Lemma~\ref{property:structure} which  complements previous results by identifying the nature of overlap between $\Phi_{kj}(.)$ and $\Phi_{\bar{k}l}(.)$ for $j\in\{1,\dots,m_k\}$ and $l\in\{1,\dots,m_{\bar{k}}\}$ for prices less than $v$. Using this lemma, we prove Property~\ref{property:difference} which is the main result of this section. 

\begin{lemma}\label{property:structure}
For every price $\tilde{p}\leq x<v$, $x$ should belong to the support sets $\Phi_{kl}(.)$ and $\Phi_{\bar{k}j}(.)$ such that $l+j>d$.
\end{lemma}

\vspace{-1mm}
A contradiction argument is used to prove the lemma. Assume that there exist $x$, $l$, and $j$ such that $x$ belongs to say $\Phi_{kl}(.)$ and $\Phi_{\bar{k}j}(.)$, and $l+j\leq d$. We show that in this case, the expected number of units sold at $x$ and $x+\epsilon$ are equal for seller $k$ when offering $l$ units, i.e. $B_{kl}(x)=B_{kl}(x+\epsilon)$,  and subsequently that $u_{kl}(x+\epsilon)>u_{kl}(x)$. Thus $x$ is not a best response for seller $k$ who offers $l$ units, which is a contradiction. 

\begin{proof}
Suppose not. There exist $x$, $l$, and $j$ such that $x$ belongs to say $\Phi_{kl}(.)$ and $\Phi_{\bar{k}j}(.)$, and $l+j\leq d$. We show that there exist $\tilde{j},\epsilon > 0$ 
 such that $x+\epsilon$ belongs in the support set of $\Phi_{\bar{k}\tilde{j}}(.)$,  and subsequently that $u_{kl}(x+\epsilon)>u_{kl}(x)$. Thus $x$ is not a best response for seller $k$ who offers $l$ units which is a contradiction.   Consider two cases:
\begin{itemize}
\item $x=\tilde{v}_{\bar{k}j}$. Using Corollary \ref{corollary:structure}, $x$ and $x+\epsilon$ belongs to the support set of $\Phi_{\bar{k},j-1}(.)$ when $\epsilon$ is small enough. Take $\tilde{j}=j-1$.
\item $x<\tilde{v}_{\bar{k}j}$. If $\epsilon$ is small enough, $x$ and $x+\epsilon$ belongs to the support set of $\Phi_{\bar{k}j}(.)$. Take $\tilde{j}=j$.
\end{itemize}

Note that since $l+j\leq d$, $l+\tilde{j}\leq d$.
We are going to argue that the expected number of units sold at $x$ and $x+\epsilon$ are equal for seller $k$, i.e. $B_{kl}(x)=B_{kl}(x+\epsilon)$. To show this, we condition on the number of available units with the seller $\bar{k}$. If  $\bar{k}$ has more than $\tilde{j}$ number of available units, say $f$,  then she will offer with price less than $x$ with probability one. Thus $\tilde{B}_{kl}(x|f)=\tilde{B}_{kl}(x+\epsilon|f)=d-f$ in which $\tilde{B}_{.}(.|.)$ is the conditional expected number of units sold. If $\bar{k}$ offers less than $\tilde{j}$ number of units, she will offer with price higher than $x+\epsilon$ with probability one. Thus  $\tilde{B}_{kl}(x|f)=\tilde{B}_{kl}(x+\epsilon|f)=l$. If $\bar{k}$ offers $\tilde{j}$ units, since $l+\tilde{j}\leq d$, $\tilde{B}_{kl}(x|\tilde{j})=\tilde{B}_{kl}(x+\epsilon|\tilde{j})=l$. Therefore the expected number of units sold at $x$ and $x+\epsilon$ are equal for seller $k$, and $u_{kl}(x+\epsilon)>u_{kl}(x)$. The proof is complete. 

\end{proof}

Finally, the following property characterizes the behavior of NE at $v$.



\begin{property}\label{property:difference} For each $k$, there exists a threshold  such that seller $k$ offers price $v$ with probability one if she has the availability level less than or equal to this threshold. We denote this threshold with $l_k$. This threshold is such that:
\begin{itemize}
\item $l_k\in\{e_k,\dots,m_k-1\}$
\item $l_1+l_2=d-1$ or $l_1+l_2=d$
\end{itemize}
The price distribution $\Phi_{kj}(.)$ does not have a jump at $v$ if $j>l_k+1$, at most one of the distributions $\Phi_{1,l_1+1}(.)$ and $\Phi_{2,l_2+1}(.)$ can have a jump at $v$, and size of such a jump
is less than $1.$ 
\end{property}


Note that in Fig~\ref{figure:dist}, $l_1=l_2=1$, and $l_1+l_2=d-1$. In addition, both sellers have a jump of magnitude one at price $v$ when they have one unit available, only seller two has a jump at price $v$ when the availability level is two, and there is no jump in the distribution functions when sellers have three units available. 

\begin{proof}                                                                                                                                                                                                                                                                                                                                 
Take $z_k$ such that  $k$ offers price $v$ with probability one if she has $i\in\{1,\dots,z_k\}$ units. Property \ref{property:e_asy} shows that $z_k \geq e_k$. We will prove that the $z_k$ should be less than $m_k$. Note that if seller $k$ has $m_k$ units of availability and she offers her units with a single price $v$, then $\tilde{p}_k=v$. By Properties \ref{property:equal_lower_bound} and \ref{property:supportsets}, the other seller, $\bar{k}$, offers her units with a single price $v$ regardless of the number of available units. This is a contradiction. The reason is because of Lemma \ref{lemma:jump_asy}. Since $m_1+m_2>d$, if $\Phi_{1,m_1}(.)$ has a jump at $v$, then $u_{2m_2}(v -\epsilon) > u_{2m_2l}(v)$, for all sufficiently small but positive $\epsilon$. Thus $v$ is not a best response for the second player when she offers $m_2$ units, which is a contradiction. Thus $z_k<m_k$. Therefore $z_k\in\{e_k,\dots,m_k-1\}$.
                                                                                                                                                                                                                                                                                                                                 

First, suppose $z_1+z_2\geq d+1$. By Lemma $\ref{lemma:jump_asy}$, $v$ is not a best response for the player $k$ when she offers $z_k$ units, which is a contradiction. Therefore $z_1+z_2\leq d$. Next, we will prove that either $z_1+z_2=d-1$ or $z_1+z_2=d$. Note that by the definition of $z_k$, seller $k$ with availability $z_k+1$ cannot choose the price $v$ with probability $1$. Thus using this fact and Corollary $\ref{corollary:structure}$, the price $x=v-\epsilon$ for $\epsilon>0$ small enough is in the support sets of $\Phi_{1,z_1+1}(\cdot)$ and $\Phi_{2,z_2+1}(\cdot)$. Thus, 
by Lemma \ref{property:structure}, $z_1+z_2\geq d-1$. Knowing that $z_1+z_2\leq d$. Take $l_k=z_k$, and the first part of the property follows.

Now we should consider the possibility of having a jump at $v$ for $\Phi_{kj}(.)$ for $j\geq l_k+1$.
  We will prove that the price distribution does not have a jump at $v$ when seller $k$ offers more than $l_k+1$ units. Suppose $\Phi_{kj}(.)$ has a jump for $j>l_k+1$. Note that $j+l_{\bar{k}}>l_k+l_{\bar{k}}+1\geq d$. By Lemma \ref{lemma:jump_asy}, $v$ is not a best response for the seller $\bar{k}$ under availability $l_{\bar{k}}$ which contradicts the definition of $l_{\bar{k}}$.

Now consider $l_k+1$. By definition of $l_k$ such a jump must have a size less than $1$, should it exist. We will prove that at most one of the distributions $\Phi_{1,l_1+1}(.)$ and $\Phi_{k,l_2+1}(.)$ can have a jump at $v$. Suppose not and both have a jump at $v$. By Lemma \ref{lemma:jump_asy}, since $(l_1+1)+(l_2+1)>d$, $v$ is not a best response for the player $k$ when she offers $l_k+1$ units. This is a contradiction. The result follows.

\end{proof}

Revisiting Equation (\ref{equation:utility}) implies that utility, $u_{ki}(.)$, is continuous not only in interval $[c,v)$, but also at price $v$, if $i\leq d-l_{\bar{k}}-1$. The reason is that for $i\leq d-l_{\bar{k}}-1$, equation (\ref{equation:utility}) depends only on $\Phi_{\bar{k}j}(.)$ where $j\geq l_{\bar{k}}+2$, which is continuous at price $v$ based on Property \ref{property:difference}. If $\Phi_{\bar{k}l_{\bar{k}}+1}(.)$ is continuous at $v$ then $u_{ki}(.)$ is continuous in $[c,v]$ for $i\leq d-l_{\bar{k}}$.

\subsection{Proof of Theorem~\ref{theorem:necessary}}
\label{summingup}

\noindent

\begin{proof}
 Part~\ref{prop:v} of Theorem~\ref{theorem:necessary} follows from Property~\ref{property:difference}. We now prove part~\ref{prop:2}. The support
set of $\Phi_{k,l_k+1}(.)$ includes at least one $x < v$ from  Property~\ref{property:difference}. Thus, Properties \ref{property:supportsets} and \ref{property:contiguous_assym} imply part~\ref{prop:2intermsupportset} of this part. Parts~\ref{prop:2intermcont} and \ref{prop:2jumpv} follow from Properties~\ref{property:continuous_as} and \ref{property:difference}, respectively.

We now prove part~\ref{prop:3}. We start with \ref{prop:2supportset}. Consider $i > l_{k+1}.$ From
Property~\ref{property:difference}, $\Phi_{k, i}(\cdot)$ does not have a jump at $v.$
From  part~\ref{prop:2intermsupportset} and Property~\ref{property:supportsets}, $v$ is not in the supports set of $\Phi_{k, i}(.)$ and $\tilde{v}_{k, i} \leq \tilde{p}_{k, i-1}.$
The result can now be proved by induction starting with $i = l_{k+2}$ using the fact that there is no gap between the support sets (Property~\ref{property:contiguous_assym}). Since $v$ is not in the support set of $\Phi_{k, i}(.)$, part \ref{prop:2cont} follows
from Property~\ref{property:continuous_as}.
Part~\ref{prop:2last}  follows from part~\ref{prop:2supportset} and Property~\ref{property:equal_lower_bound}.

Part \ref{prop:equal} follows from the fact that every price in the support set of a NE, except those on the boundaries, should be a best response for a seller. Thus they yield the same utility value. The result follows for the boundary points  of the support sets other than $v$ from Property~\ref{property:continuous_as}.
\end{proof}

\section{Arbitrary Demand}
\label{section:duopolysufficiency}

Note that the existence of the mixed strategy NE follows from Corollary 5.2 in \cite{existence}. In this section, first we present the sufficiency theorem for $d\geq \max\{m_1,m_2\}$ (Theorem~\ref{theorem:necessity_random}). Theorem~\ref{theorem:necessity_random} establishes that a strategy profile which satisfies the mentioned properties in Theorem \ref{theorem:necessary} constitutes an NE when  $d\geq \max\{m_1,m_2\}$. Note that unlike Theorem~\ref{theorem:necessary}, the sufficiency theorem holds even when $d=\max\{m_1,m_2\}$. Thus, the properties in Theorem~\ref{theorem:necessary} are both necessary and sufficient conditions for an NE when $d>\max\{m_1,m_2\}$, and only sufficient conditions when  $d=\max\{m_1,m_2\}$. The sufficiency theorem naturally leads to an algorithm for computing NE strategy profiles that satisfy the properties in Theorem~\ref{theorem:necessary} (Appendix~\ref{section:duopolyalgorithm_asym_framework}). Any strategy profile obtained by the algorithm constitutes an NE by Theorem~\ref{theorem:necessity_random}. In Section~\ref{easygeneralizationduopoly}, we argue that the computation of the NE strategies for $d<\max\{m_1,m_2\}$ can be reduced to $d=\max\{m_1,m_2\}$. This completes the entire framework.

\subsection{The Sufficiency Theorem  when $d\geq\max\{m_1,m_2\}$} \label{subsection:sufficiencydgm}

\begin{theorem}\label{theorem:necessity_random}
Consider a strategy profile that satisfies the properties enumerated in Theorem \ref{theorem:necessary}. This strategy profile is a Nash equilibrium when  $d\geq \max\{m_1,m_2\}$. 
\end{theorem}

The proof is presented in Appendix \ref{appendix:sufficiency}.  In the proof, we use the fact that $A_{k,l,j}(.)$ is non increasing and non positive when  $d\geq \max\{m_1,m_2\}$.


\subsection{Allowing $d\leq \max\{m_1,m_2\}$}
\label{easygeneralizationduopoly}

Note that all results before equation \eqref{equation:utility} also hold when $d\leq \max\{m_1,m_2\}$. Thus \eqref{equation:utility} can be restated by replacing $e_k=d-m_{\bar{k}}$ with $e_k=(d-m_{\bar{k}})^+$:



\small
\vspace{-4mm}
\begin{equation}\label{equation:utility_generalized}
\begin{aligned}
B_{kj}(x)=j\sum_{i=0}^{(d-j)^+}{q_{\bar{k}i}}&+\min\{j,d\}\sum_{i=(d-j)^++1}^{m_{\bar{k}}}{\big{(}1-\Phi_{\bar{k}i}(x)\big{)}q_{\bar{k}i}}\\
&+\sum_{i=(d-j)^++1}^{m_{\bar{k}}}{\Phi_{\bar{k}i}(x)q_{\bar{k}i}(d-i)^+}
\end{aligned}
\end{equation}

\normalsize
Note that if $m_{k}>d$, the utilities of all number of availability levels $j\geq d$ for player $k$ are equal:

\vspace{-3mm}
\begin{equation}
\begin{aligned}
&u_{kd}=u_{k,d+1}=\dots=u_{km_k}=d\sum_{i=1}^{m_{\bar{k}}}{\big{(}1-\Phi_{\bar{k}i}(x)\big{)}q_{\bar{k}i}}
\end{aligned}
\end{equation}

Let $\tilde{q}_{\bar{k}d}=\sum_{i=d}^{m_{\bar{k}}}{q_{\bar{k}i}}$ and $\tilde{\Phi}_{\bar{k}d}(x)=\sum_{i=d}^{m_{\bar{k}}}{\frac{q_{\bar{k}i}}{\tilde{q}_{\bar{k}d}}\Phi_{\bar{k}i}(x)}$. Thus, $\tilde{q}_{\bar{k}d}$ is the probability that the availability level of seller $\bar{k}$ is greater than or equal to $d$ and  $\tilde{\Phi}_{\bar{k}d}(x)$ is the average probability distribution associated with selecting the price if seller $\bar{k}$ availability is $d$ or higher. Now, the term $\sum_{i=d}^{m_{\bar{k}}}{\big{(}1-\Phi_{\bar{k}i}(x)\big{)}q_{\bar{k}i}}$ in the expression for $u_{ki}(.)$ in  (\ref{equation:utility_generalized}) can be replaced by $\tilde{q}_{\bar{k}d}(1-\tilde{\Phi}_{\bar{k}d}(x))$. Thus the problem is reduced to finding the structure when $d=\max\{m_1,m_2\}$.  It was proved previously that a strategy profile that satisfies properties in Theorem~\ref{theorem:necessary} is a NE when $d= \max\{m_1,m_2\}$. Thus, a set of equilibria of the game when $d<\max\{m_1,m_2\}$ can be found by defining $\tilde{\Phi}_{kd}(.)$ and using the properties in Theorem~\ref{theorem:necessary}. The distribution of each individual $\Phi_{kj}(.)$ for $j\geq d$ cannot be determined uniquely and is not of significant interest.

\section{The Symmetric Setting}
\label{section:duopolyalgorithm}

We now consider the symmetric setting in which $\vec{q_1}=\vec{q_2}=\vec{q}$ (clearly $m_1=m_2=m$). In this case, it is natural to consider a symmetric NE, defined as follows,
 \begin{definition}
An NE $(\Theta_1(\cdot),\Theta_2(\cdot))$ is said to be \emph{symmetric} if $\Theta_1(\cdot)=\Theta_2(\cdot)$. 
 \end{definition}
Thus, when considering symmetric NE, in terminologies like $\Phi_{.}(\cdot), \Theta_{.}(\cdot), u_{.}(\cdot), \tilde{p}_{\cdot}$, we drop the index that represents the seller and only retain the index that represents the number of units available for sale. As a special case of the general setting (Sections~\ref{section:asymmetric_duopoly} and \ref{section:duopolysufficiency}), every symmetric NE should satisfy the properties in Theorem~\ref{theorem:necessary} when $d> m$ , and every strategy profile that satisfies these properties is a NE when $d\geq m$ (Theorem~\ref{theorem:necessity_random}).  In Section~\ref{section:sym_prop}, we extend Theorem~\ref{theorem:necessary} to the case of $d=m$. In Section~\ref{section:symm_alg}, we will present an algorithm to find symmetric Nash equilibria of the game when $d\geq m$. Using the results in Section~\ref{easygeneralizationduopoly}, the algorithm can be extended to $d<m$.  

Note that the algorithm reveals that there is only one symmetric strategy profile that satisfies the properties. It follows from  Theorems \ref{theorem:necessary} and \ref{theorem:necessity_random} that a symmetric NE strategy profile  uniquely exists when $d\geq m$. In contrast, in  Appendix \ref{section:examples}, we show that there may exist multiple Nash equilibria for an asymmetric market. It is not
clear if there exists an asymmetric NE for the symmetric market; our extensive numerical evaluations have not however led such strategy.

\subsection{Properties of a Symmetric Nash Equilibrium}\label{section:sym_prop}

\begin{theorem}\label{theorem:symmetricsuff}
Let $d=m$. A symmetric NE in a symmetric market satisfies the properties in Theorem~\ref{theorem:necessary}. 
\end{theorem}

The proof is technical and is relegated to the Appendix. It implies that properties in Theorem~\ref{theorem:necessary} are necessary and sufficient conditions for a symmetric NE when $d\geq m$. 

Since NE is symmetric, $l^*=l_1=l_2$. Thus,  $l^*=\frac{d-1}{2}$ or $l^*=\frac{d}{2}$, whichever is an integer. Since at most one seller can have a jump at $v$ at $l^*+1$, in a symmetric NE, none of them do. Thus, the
 properties in Theorem \ref{theorem:necessary} transform to the following in the symmetric context.

\begin{enumerate}
\item \label{prop:S1} Sellers offer price $v$ with probability $1$, if they have $i\in\{1,\dots,l^*\}$ available units.
\item  There exists an increasing sequence $a_{m}, a_{m-1}, \ldots, a_{l^*+1}, a_{l^*}$ of positive real numbers in  $(c, v]$ with $a_{l^*} = v$  such that each seller randomizes her price in the interval $[a_{i},a_{i-1}]$ when she has $i$ units of commodity available for sale for $i\in\{l^*+1,\dots,m\}$. Thus,
\begin{enumerate}
\item \label{prop:SCont} Support sets are contiguous.
\item \label{prop:Sdisj} Support sets are disjoint (except possibly at one point).
\item \label{prop:Sdec} Support sets are in decreasing order of the number of available units for sale.
\end{enumerate}
\item \label{prop:Scont} Price distribution is continuous for $i\geq l^*$.
\item \label{prop:Sequal} The utility of a seller when she offers $i$ units is equal for all  prices in the support set of $\Phi_{i}(.)$, except possibly at price $v$ (if it belongs to her support set).
\end{enumerate}

\subsection{Algorithm for computing a symmetric  NE for the symmetric setting}\label{section:symm_alg}

We will now identify an algorithm to compute strategies that exhibit the  properties in the previous subsection. The algorithm reveals that there is only one symmetric strategy profile that satisfies the same. It follows from  Theorem \ref{theorem:necessary} and \ref{theorem:necessity_random} that a symmetric NE strategy profile uniquely exists when $d\geq m$. Note that the algorithm is developed for $d\geq m$. However, with the method presented in Section \ref{easygeneralizationduopoly}, the algorithm can be used to find the equilibrium for $d\leq m$.

Since $\Phi_j(\cdot)$ is completely characterized for $j < \frac{d+1}{2}$, we should characterize  $\Phi_j(\cdot)$ for $j \geq \frac{d+1}{2}$, and outline a framework for computing the same. We proceed in an increasing order of $j$ starting  with $j = \lceil\frac{d+1}{2}\rceil$. Then moving to $j=\lceil\frac{d+1}{2}\rceil+1$, etc.

Now, let $\lceil\frac{d+1}{2}\rceil$. Note that $\tilde{v}_{\lceil\frac{d+1}{2}\rceil} = v$ and $\tilde{p}_k = v$ for $k <\lceil\frac{d+1}{2}\rceil$,  and $\tilde{v}_k \leq \tilde{p}_{\lceil\frac{d+1}{2}\rceil}$ for $k >\lceil\frac{d+1}{2}\rceil$ (Properties~\ref{prop:S1} and \ref{prop:Sdec}). Since support sets are ordered (Property~\ref{prop:Sdec}) and disjoint (Property~\ref{prop:Sdisj}), the expression for $u_{\lceil \frac{d+1}{2}\rceil}(x)$ for $x\in [\tilde{p}_{\lceil\frac{d+1}{2}\rceil},v)$ only depends on $\Phi_{\lceil\frac{d+1}{2}\rceil}(x)$(Equation~\eqref{equation:utility}). In particular,
  $u_{\lceil \frac{d+1}{2}\rceil}(v^{-})$ can be obtained using the fact that
  $\Phi_{\lceil\frac{d+1}{2}\rceil}(v^{-}) = 1$ which follows from the continuity of $\Phi_{\lceil\frac{d+1}{2}\rceil}(.)$
  (Properties~\ref{prop:Scont}).
    Next, $u_{\lceil \frac{d+1}{2}\rceil}(x)=u_{\lceil \frac{d+1}{2}\rceil}(v^{-})$ for every $x\in [\tilde{p}_{\lceil \frac{d+1}{2}\rceil}, v)$.
Thus having $u_{\lceil \frac{d+1}{2}\rceil}(v^{-})$, and using continuity, we can find a unique expression for $\Phi_{\lceil\frac{d+1}{2}\rceil}(x)$. Using $\Phi_{\lceil \frac{d+1}{2}\rceil}(\tilde{p}_{\lceil \frac{d+1}{2}\rceil})=0$, $\tilde{p}_{\lceil \frac{d+1}{2}\rceil}$ can be found uniquely.

We now compute the structure of $\Phi_i(\cdot)$, $\forall i > \lceil\frac{d+1}{2}\rceil$ using $\Phi_{i-1}(.), \Phi_{i-2}(.),\cdots, \Phi_{\lceil \frac{d+1}{2}\rceil}(.)$ that are computed before $\Phi_i(\cdot)$. We utilize the facts that,

\begin{enumerate}
\item $\Phi_j(x) = 1$ for $j > i$, $x \in [\tilde{p}_i, \tilde{v}_i]$
\item  $\Phi_j(x) = 0$ for $j < i$, $x \in [\tilde{p}_i, \tilde{v}_i]$
\item $\tilde{v}_i < v$
\end{enumerate}

Thus, from (\ref{equation:utility}),
\begin{equation}\label{equ:payval}
u_i(\tilde{v}_i) =(\tilde{v}_i-c)\bigg{(}i\sum_{g=0}^{i-1}{q_g}
+\sum_{i}^m{q_g(d-g)}\bigg{)}
\end{equation}
Since $\tilde{v}_i = \tilde{p}_{i-1}$, and $\tilde{p}_{i-1}$ is computed during the computation of $\Phi_{i-1}(\cdot)$, which precedes that of $\Phi_{i}(\cdot)$, (\ref{equ:payval}) fully specifies $u_i(\tilde{v}_i)$. Furthermore, for $x \in [\tilde{p}_i,\tilde{v}_i]$ the only unknown variable in the expression of $u_{i}(x)$ is $\Phi_{i}(x)$. Since $u_i(x) = u_i(v_i)$ for $x \in [\tilde{p}_i,\tilde{v}_i]$,
\begin{equation}\label{equ:dist2}
\Phi_{i}(x)=\frac{i\sum_{g=0}^{i-1}{q_g}+i q_{i}+\sum_{g=i+1}^m{q_g(d-g)}-\frac{u_i(\tilde{v}_i)}{x-c}}{q_{i}(2i-d)}
\end{equation}

From (\ref{equ:dist2}), $\Phi_i(\tilde{v}_i) = 1$. Thus, for $x \geq \tilde{v}_i$, $\Phi_i(x) = 1$. Now, $\tilde{p}_i$ can be uniquely identified using the fact that $\Phi_i(\tilde{p}_i) = 0$,

\vspace{-18pt}
\begin{equation}\label{equ:lower2}
\tilde{p}_{i}=c+\frac{(\tilde{v}_i-c)\bigg{(}i\sum_{g=0}^{i-1}{q_g}
+\sum_{i}^m{q_g(d-g)}\bigg{)}}{i\sum_{g=0}^{i-1}{q_g}+i q_{i}+\sum_{g=i+1}^m{q_g(d-g)}}
\end{equation}
Therefore $\Phi_i(x) = 0$ for $x \leq \tilde{p}_i$. Clearly, $\Phi_i(\cdot)$ has been characterized uniquely. Note that the denominator of \eqref{equ:lower2} is positive since $d\geq m$ and $q_m<1$ (uncertainty assumption in Section~\ref{system_model}).  In addition, $\tilde{p}_{i}>c$. This is because of the fact that the second term of RHS of (\ref{equ:lower2}) is positive.


We now prove that $\Phi_i(\cdot)$ is a valid probability distribution. Clearly, $\Phi_{i}(\cdot)$ is continuous. Note that in (\ref{equ:dist2}) for $x\in [\tilde{p}_{i}, \tilde{v}_i)$, by increasing $x$, the term $\frac{u_{i}(v_i)}{x-c}$ will strictly decrease (since $u_{i}(\tilde{v}_i)>0$), and we can say that $\Phi_{i}(x)$ is strictly increasing. Also, $\Phi_{i}(\tilde{p}_{i}) = 0$ and $\Phi_{i}(\tilde{v}_i) = 1$. Thus, $0 \leq \Phi_{i}(x) \leq 1$ for $x \in [\tilde{p}_{i},\tilde{v}_i)$. Therefore, $\Phi_{i}(\cdot)$ is non-decreasing and assumes values in $[0, 1]$ for all $x$. The claim follows. Thus we have uniquely identified a symmetric strategy that satisfies the properties required by a Nash equilibrium.


\section{Random Demand}\label{section:random}
We have so far assumed that the  demand $d$ is deterministic. In this section, we will generalize the results to a random demand, $D.$ Let $r_d$ denote the probability that the demand is $d$, $B_{kld}(x)$ be the expected number of units that seller $k$ sells if she offers $l$ units for sale and quotes $x$ as the price per unit when the total demand is $d$, and $u_{kld}(x)$ be the expected utility in this case. Clearly,
\begin{equation}
u_{kl}(x)=\sum_{d}{r_du_{kld}(x)}=\sum_{d}{r_d B_{kld}(x)(x - c)}
\nonumber
\end{equation}
We introduce $\underline{d}=\min\{d: d > 0\ and\  r_d > 0\}$. Utilizing similar proofs, we can show that all the previous results about the structure of NE are valid for the random demand, once $d$ is replaced with $\underline{d}$. This is but expected as each seller now chooses her price knowing that she is assured of an overall demand of at least $\underline{d}$ (instead of $d$ in the deterministic demand case).  Algorithms similar to those in the deterministic case can be developed for computation of the NE in both symmetric and general cases.

\section{The Asymptotic Behavior}\label{section:simulations}

In this section, through numerical evaluations, we investigate the asymptotic behavior of the symmetric NE of a symmetric duopoly market when the number of available units with a seller increases to infinity. In asymptotic scenario, many of availability probability distributions that arise naturally concentrate around the mean. Thus, $q_k\rightarrow 0$, when $k$ is far from the mean.    First, we show  that the length of the support set for availability of $k$ units approaches zero as $q_k\rightarrow 0$: From equation~\eqref{equ:lower2},

\vspace{-7pt}
$$
\begin{aligned}
\tilde{p}_i&=c+\frac{(\tilde{p}_{i-1}-c)(i\sum_{g=0}^{i-1}{q_g}+\sum_{g=i}^m{q_g(d-g)})}{i\sum_{g=0}^{i}{q_g}+\sum_{g=i+1}^m{q_g(d-g)}}\\
&=\tilde{p}_{i-1}+(\tilde{p}_{i-1}-c)\frac{q_i(d-2i)}{i\sum_{g=0}^{i}{q_g}+\sum_{g=i+1}^m{q_g(d-g)}}
\end{aligned}
$$
It is immediate that if $q_i\rightarrow 0$, then $\tilde{p}_{i}\rightarrow \tilde{p}_{i-1}$\footnote{Note that the denominator is positive since $d\geq m$, and we assume uncertainty in competition, i.e. $q_m<1$.}. This implies that the length of the support set for the availability level $i$ units approaches zero.

We investigate the asymptotic behavior using numerical simulations when the availability of each seller follows a binomial distribution $(m,r<1)$. With this distribution, as $m\rightarrow \infty$,  the binomial distribution can be approximated by a normal distribution with mean $mr$ and variance $mr(1-r)$. Thus $m\rightarrow \infty$ yields that $\tilde{p}_i\rightarrow \tilde{p}_{i-1}$ when $|i-mr|$ is large enough. In other words, the length of the support set for the availability level $i$ units approaches zero if $i$ is far from the mean. Other parameters are considered to be  $v=10$, $c=1$, and $d=m$.

\begin{figure}[t]
\begin{center}
\includegraphics[scale=0.32]{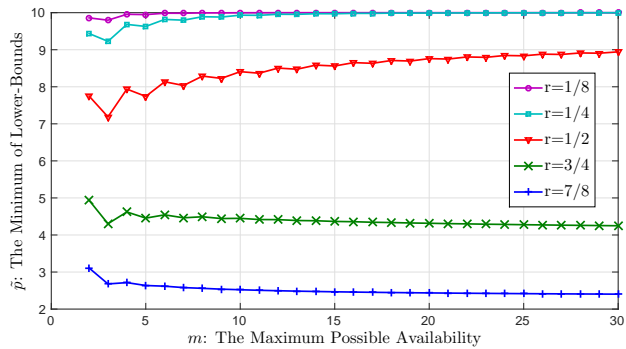}
\end{center}
\caption{$\tilde{p}$ versus $m$ for when availability level is binomial with probability $p$ and demand is $m$}
\label{figure:tilde_p}
\end{figure}

In Figure~\ref{figure:tilde_p}, the value of  $\tilde{p}$, i.e. the lowest lower-bound is plotted versus $m$, i.e. the highest possible level of availability. As you can see, the larger the probability $r$, the smaller $\tilde{p}$. Note that when $r$ is large, the seller is more likely to offer with higher levels of availability. Therefore the competition is more intense. In addition, when $m$ is increased, the distribution $\vec{q}$ of the availability levels concentrates around the mean, $mr$. If $r > \frac{1}{2}$, when a seller offers $k=mr$, knowing that the other seller offers $mr>\frac{m}{2}$ with positive probability, she will offer price less than $v$ (note that $d=m$). Furthermore, the higher $m$, the more intense the competition, and consequently $\tilde{p}$ is decreasing. On the other hand, when $r\leq \frac{1}{2}$, if a seller offers  around $mr$ units, there is no competition between sellers knowing that $2mr\leq d=m$. Furthermore, the availability probability $q_k$, when $k$ is far from $mr$, tends to zero when $m$ is large. Thus the associated support sets shrink to zero. This explains the increasing behavior of $\tilde{p}$. We notice oscillation in the figure, since $m$ alternates between odd and even.

\section{Applications and Discussion}\label{application}

The framework we described in this paper can be used to model three different applications in which uncertainty in competition naturally emerges: secondary spectrum access networks, a non-neutral Internet market, and micro grid networks.

Pricing in secondary spectrum access networks \cite{CRN} is one of the applications of our model. Recent developments in wireless devices have resulted in a significant growth in demand for the radio spectrum. This leads to spectrum congestion. On the other hand, the available radio spectrum is greatly under-utilized \cite{survey}. Spectrum congestion and under-utilization have directed researchers to adopt new techniques in order to use the available spectrum more efficiently and to decrease congestion. Secondary spectrum access is an example of these techniques. In these networks, there are two types of users: (i) Primary/licensed users, who lease a number of frequency bands (channels) directly from the regulator, and (ii) Secondary/unlicensed users, who lease frequency bands from primary users for a certain amount of time in exchange for money or other types of credit. Note that primary and secondary users correspond to sellers and buyers in our model, respectively. Each primary user may have multiple vacant frequency bands available for sale, and a secondary user can lease a channel only if it is not in use by the primary user who owns it. The usage of subscribers of primary users is random and different for different primaries.  Thus primaries are uncertain about the competition, and they need to select prices for the frequency bands they offer for sale, without knowing the number of frequency bands available for sale with their competitors. 

The next application of our model is the interaction between Internet Service Providers (ISPs) and Content Providers (CPs) in a \emph{non-neutral} Internet. Net neutrality on the Internet is perceived as the policy that mandates ISPs to treat all data equally, regardless of the source, destination, and type of the data \cite{progressive}. This precludes ISPs from charging CPs to carry their data to the end-users in the last-mile. In January 2014, a federal appeals court struck down parts of the Federal Communication Commission's (FCC) rules for Net-Neutrality \cite{NYtimes1}. The new changes in the Internet policies enable ISPs to change their policies and provide differential treatment of traffic to generate additional revenue streams from CPs. This is called a \emph{non-neutral} regime for the Internet in which ISPs can offer resources to CPs for sale or rent. Our framework can capture the pricing in this type of the Internet market. Here, ISPs represent the sellers that offer resources for reservation/sponsorship, and CPs are the customers that shop around for the lowest available prices. Therefore ISPs seek to set prices that will ensure that their resources are reserved/sponsored and also fetch adequate profit. Note that ISPs determine the number of resources available for sponsoring based on the demand of their end-users. The more congested an ISP, the higher the demand of end-users, and therefore the lower the number of available resources for sponsoring. Since the demand of end-users is not a priori known, the ISPs are not aware of the number of units of resources available to her competitor before quoting her price. Thus, the competition that each ISP faces is uncertain.

The third example scenario pertains to pricing in micro grids\cite{microgrids2}. A micro grid network is a network of distributed power generating systems connected to local subscribers, and also to the central macro power grid. The distributed generation of power at small on-site stations is a promising alternative to the traditional generation at large stations. Decreasing the loss of transmission by reducing the distance to consumption units \footnote{In microgrid networks, the power can be sold to or bought from other local micro grids. This reduces the distance the power should be transmitted via the macro grid from a generation to a consumption site.}, utilizing renewable energy sources, decreasing the risk of blackout, and increasing security are some of the advantages of distributed power generating scheme \cite{microgrids}.  In these networks, a microgrid equipped with a distributed power generating system can sell its excess power to other microgrids as well as the macro grid. Since micro grids are emerging technologies \footnote{Microgrids are emerging  in different countries such as United States \cite{US_micro} and India\cite{India_micro}.}, their  market structure has not been finalized yet. Thus, different market structures needs to be investigated. One possible scenario is a centralized market in which micro grids sell their excess power to the macro grid or a local utility at a feed-in tariff \footnote{A feed-in tariff is an offer by the macro grid to purchase some or all of the output of a micro grid at a fixed or formula rate.}\cite{feed-in}. Another scenario, which is investigated in this paper, is a distributed market in which micro grids trade the power among themselves as also with macro grid at a price quoted by them in a competitive market. Our model captures the second scenario in which each micro grid  with excess power (seller) sells its excess power to micro grids with deficient power or the macro grid (buyers)\footnote{Note that each microgrid can be a seller or a buyer depending on the number of power units generated and the demand of its subscribers. However, at a fixed time, the identity of a micro grid as a seller or a buyer is fixed.}. The amount of power generated by a power generating system is not a priori known and is different for different sellers. Thus, the sellers need to select prices for the excess power they offer for sale, without knowing  the number of power units available for sale with their competitors (uncertainty in competition).



Note that in these applications, we considered the case that sellers sell their surplus supply. The original supply is allocated to their subscribers, i.e. contracted customers, using either usage-based or flat-rate pricing.

We now discuss about some details of the applications that arise in practice. Note that one unit of commodity might be valued differently by different buyers in the above mentioned applications. For instance, different secondary users receive different rates for the same frequency band, depending on their location.  Similarly, different microgrids receive different amounts of power owing to differences in power loss.  Hence, different buyers have different utilities even when they buy the same amount of commodity. However, in our formulations, we assumed that the pricing structure is the same for all buyers, regardless of the differences in the utilities.   We justify this assumption as follow. 

First note that in microgrid networks, the  transmission loss is typically negligible, due to the proximity of generators and consumers. Thus, all consumers receive approximately the same utility for a unit of power they purchase.   For Primary/Secondary markets and a Non-Neutral Internet market, the utility of secondary users and CPs (as buyers) depends on the utility of their end-users, and subsequently is different for different secondaries and CPs, depending on the characteristics of their end-users. Sellers would not in general know the characteristics and identities of the subscribers of potential buyers. Hence, prices quoted by the sellers cannot depend on the utility of buyers. In addition,  note that introducing a \emph{differential pricing} for customers complicates the  pricing structure for them, and prevents an easy cost prediction and management. For instance, in wireless settings, the channel quality of end-users and the rate perceived by them are time and location dependent\cite{ISIT}. Thus, in a differential pricing scheme, customers know the current pricing only when they use the service. But, customers are usually reluctant to adopt differential pricing schemes, owing to the rapid variability of prices which is not usually well-received by them \cite{soumya_survey}. In addition, sellers are also reluctant using a differential pricing scheme for their end-users, as they are usually computationally complex. Therefore, we did not consider different valuations for different customers in determining the pricing strategy of sellers. However, differential pricing for users with different valuation might arise for other applications; this constitutes a topic of future research.

\section{Conclusion}
We investigated  price competition in a duopoly market with uncertain competition when different sellers may have different number of units available for sale. We modelled the interactions among sellers as a non-cooperative game and listed a set of properties that are sufficient conditions for a strategy profile to be an NE. We proved that these properties are also necessary conditions for an NE in a symmetric market, or for some values of demand values in an asymmetric market. We showed that a symmetric NE uniquely exists and presented an algorithm for computing the same. In Appendix~\ref{appendix:oligopoly}, using the results proved for a duopoly, we proposed a heuristic pricing strategy for sellers in a symmetric oligopoly market. Numerical results reveal that the proposed pricing strategies are good approximations of NE when sellers are not too concerned about optimizing over small gains. A Direction for future work is to consider different pricing for different types of demand.

\appendices

\section{Proof of Lemma~\ref{lemma:jump_asy}}\label{appendix:lemma1}

\vspace{-3pt}
\begin{proof}
First consider the tuple $<l,y>$ associated with the seller $\bar{k}$ in which the first element is the number of units she offers and the second one is the price she chooses. We introduce $D_{kl}^{(1)}(y,i,x)$ as the expected number of units sold by the seller $k$ who wants to offer $l$ units with price $y$ when her competitor's tuple $<g,z>\neq <i,x>$, and $D_{kl}^{(2)}(y,i,x)$ as the expected number of units sold by the seller who wants to offer $l$ units with price $y$ when her competitor's tuple $<g,z>=<i,x>$. The expected number of units sold by a seller can be written as,

\vspace{-3mm}
\begin{equation}
\begin{aligned}
B_{kl}(y)&=D_{kl}^{(1)}(y,i,x)Pr\{<g,z>\neq <i,x>\}\\
&+D_{kl}^{(2)}(y,i,x)Pr\{<g,z>=<i,x>\}
\end{aligned}
\nonumber
\end{equation}

Note that $D_{kl}^{(1)}(a,i,x)\leq D_{kl}^{(1)}(x,i,x)$ and $D_{kl}^{(2)}(a,i,x)\leq D_{kl}^{(2)}(x,i,x)$ for $a\geq x$ because the number of units a seller sells is a non-increasing function of her price for any given amounts offered by both sellers and any given price chosen by the competitor. Thus $B_{kl}(a)\leq B_{kl}(x)$. In addition, 

\small
\vspace{-3mm}
\begin{equation}\label{equ:non_vanishing2}
\begin{aligned}
&B_{kl}(x-\epsilon')-B_{kl}(x)=(D_{kl}^{(1)}(x-\epsilon',i,x)\\
&\qquad-D_{kl}^{(1)}(x,i,x))Pr\{<g,z>\neq <i,x>\}\\&\qquad +(D_{kl}^{(2)}(x-\epsilon',i,x)-D_{kl}^{(2)}(x,i,x))Pr\{<g,z>=<i,x>\}
\end{aligned}
\end{equation}

\normalsize
As we discussed $D_{kl}^{(1)}(x,i,x)\leq D_{kl}^{(1)}(x-\epsilon',i,x)$. For $D_{kl}^{(2)}(x,i,x)$, we should consider ties. Since each buyer is equally likely to buy a unit from both sellers if both select equal prices, we can say that $D_{kl}^{(2)}(x,i,x)=l\frac{d}{i+l}<l$ (since $i+l>d$) and $D_{kl}^{(2)}(x-\epsilon,i,x)=l$. Note that $Pr\{\text{other seller's tuple}<g,z>=<i,x>\}=q_i\times \text{Jump Size of $\Phi_{ki}(.)$ at $x$}$. Thus, for all positive $\epsilon'$, RHS of (\ref{equ:non_vanishing2}) is greater than or equal to $\theta(x)$, where $\theta(x)$ is a positive number that does not depend on $\epsilon$. Therefore since $B_{kl}(a) \leq B_{kl}(x)$, $\forall a \geq x$, $B_{kl}(x-\epsilon')\geq  B_{kl}(a)+ \theta(x)$, for all $a\geq x$. Thus,

\vspace{-3mm}
\begin{equation}
u_{kl}(x-\epsilon')-u_{kl}(a)\geq (x-\epsilon'-a){B_{kl}(a)} + \theta(x)(x-\epsilon'-c)
\nonumber
\end{equation}

Since $x > c$, for all sufficiently small $\epsilon'$, $ x-\epsilon'-c > 0$. In addition, since $a\leq x+\epsilon$ by the statement of the lemma, the lowest value for $x-\epsilon'-a$ is $-\epsilon-\epsilon'$, and ${B_{kl}(a)}\leq l$. Therefore $(x-\epsilon'-a){B_{kl}(a)} + \theta(x)(x-\epsilon'-c)\geq (-\epsilon-\epsilon')l+\theta(x)$.  Therefore, for all sufficiently small but positive $\epsilon$ and $\epsilon'$,

\vspace{-2mm}
\begin{equation}
u_{kl}(x-\epsilon')>u_{kl}(a) \qquad a\in [x,\min\{x+\epsilon,v\}]
\nonumber
\end{equation}
\end{proof}

\section{Proof of Theorem~\ref{theorem:necessity_random}}\label{appendix:sufficiency}
\begin{proof}
The goal is to show that for each $i$ and $k$ all $x \in [\tilde{p}_{ki}, \tilde{v}_{ki})$ constitutes a best response for the seller $k$ who offers $i$ units. That is, for each $x \in [\tilde{p}_{ki}, \tilde{v}_{ki})$ and for all $y$, $u_{ki}(x) \geq u_{ki}(y)$. In addition, if $\Phi_{ki}(\cdot)$ associates positive probability with $\tilde{v}_{ki}$, then  $u_{ki}(\tilde{v}_{ki}) \geq u_{ki}(y)$ for all $y$, i.e., $v_{ki}$ is a best response when the seller $k$ offers $i$ units.
Note that the distributions, $\Phi_{ki}(\cdot)$'s, should satisfy Property \ref{property:continuous_as}. Thus, equations (\ref{equation:utility}) and (\ref{equation:diff_decreasing}) holds for $x<v$, and $A_{k,l,j}(x)$ is non increasing and non positive with respect to $x$ for $l>j > e_{\bar{k}}$. 

We consider the case $j \leq e_{\bar{k}}$ here. Thus, $B_{k,j}(x) = j$ and $B_{k,l,j}(x) = \frac{1}{l}B_{k,l}(x) - 1$.
Note that the expected number of units $ B_{k,l}(x)$ sold at price $x$ when $l$ units are offered  is a non-increasing function of $x$ and $ B_{k,l}(x) \leq l.$ Thus, $B_{k,l,j}(x)$ and therefore $A_{k,l,j}(x)$ is non increasing and non positive with respect to $x$ for $l>j$ regardless of how $j$ compares with $e_{\bar{k}}$.

Consider $x < \tilde{p}$. $u_{ki}(x) \leq i(x-c) < i(\tilde{p} - c) = u_{ki}(\tilde{p})$. The last equality follows from (\ref{equation:utility}), since $\Phi_{\overline{k}j}(\tilde{p}) = 0$ for all $j$. Therefore we consider $x\geq \tilde{p}$ throughout the proof.


Suppose $l_k \in \{0,1,\dots,m_k-1\}$ in Property \ref{property:difference} is fixed. We first start with $i \geq l_k+1$. From the assumption in Theorem~\ref{theorem:necessity_random}, we know that $u_{ki}(x) = u_{ki}(y)$ for any $x, y$ in
the interior of the support set of $\Phi_{ki}(\cdot)$, the support set of $\Phi_{ki}(\cdot)$ is $[\tilde{p}_{ki}, \tilde{v}_{ki}]$, $\Phi_{ki}(\cdot)$ is continuous at all $x < v$, $\tilde{v}_{ki} < v $ for $i>l_k+1$, and $\tilde{v}_{ki} = v $ for $i=l_k+1$. Thus,  if $i>l_k+1$ $u_{ki}(x) = u_{ki}(y)$ for all $x, y \in [\tilde{p}_{ki}, \tilde{v}_{ki}]$, and for $i = l_k+1$, $u_{ki}(x) = u_{ki}(y)$ for all $x, y \in [\tilde{p}_{ki}, \tilde{v}_{ki})$. We consider the last case in detail.  Here, $\tilde{v}_{ki} = v$. If $\bar{k}$ has a jump at $v$ when she offers $l_{\bar{k}}+1$ units, by Lemma \ref{lemma:jump_asy},  $u_{ki}(v) < u_{ki}(v-\epsilon)$ for arbitrary small but positive $\epsilon$. \footnote{Note that Lemma \ref{lemma:jump_asy} holds for any arbitrary price distributions and not only those that are NE.} If not, using equation (\ref{equation:utility}) and continuity of the price distributions included in that equation, it follows that $u_{ki}(v)=u_{ki}(\tilde{p}_{ki})$. Thus, we only need to prove that for all $x$, $u_{ki}(\tilde{p}_{ki}) \geq u_{ki}(x)$. We do so by separately 
considering three cases: 1. $i\geq l_k+1$ and $x \in [\tilde{p}, \tilde{p}_{ki})$ 2. $i\geq l_k+1$ and $x \in (\tilde{v}_{ki}, v]$ 3. $i\leq l_k$.

1) {$i\geq l_k+1$ and $x \in [\tilde{p}, \tilde{p}_{ki})$}: The claim follows by vacuity for $i = m_k$. We therefore consider $i < m_k$. Since $\tilde{v}_{kj} = \tilde{p}_{k,j-1}$ for $j\geq l_k+1$, any such $x$ is in $[\tilde{p}_{kg}, \tilde{p}_{k,g-1})$ for some $g > i$. We prove this claim by induction on $g$, starting with the base case of $g = i+1$. For $x\in[\tilde{p}_{k,i+1},\tilde{p}_{ki})$,
\begin{equation}
\begin{aligned}
\frac{1}{i+1}u_{k,i+1}(x)-\frac{1}{i}u_{ki}(x)&=A_{k,i+1,i}(x)\\
\frac{1}{i+1}u_{k,i+1}(\tilde{p}_{ki})-\frac{1}{i}u_{ki}(\tilde{p}_{ki})&=A_{k,i+1,i}(\tilde{p}_{ki})\\
u_{k,i+1}(x)&=u_{k,i+1}(\tilde{p}_{ki}) 
\end{aligned}
\nonumber
\end{equation}
Note that  $\tilde{p}_{ki} = \tilde{v}_{k,i+1}$. Subtracting the first and the second equation, we get,

\footnotesize
\begin{equation}
\frac{1}{i}(u_{ki}(x)-u_{ki}(\tilde{p}_{ki}))=A_{k,i+1,i}(\tilde{p}_{ki})-A_{k,i+1,i}(x)\leq0
\nonumber
\end{equation}

\normalsize
 Since $A_{k,l,j}(x)$ is non increasing and non positive with respect to $x$ for $l>j$. Therefore $u_{ki}(x)\leq u_{ki}({\tilde{p}_{ki}})$ for $x \in [\tilde{p}_{k,i+1},\tilde{p}_{ki})$. We want to prove that $u_{ki}(x)\leq u_{ki}(\tilde{p}_{ki})$ for $x\in [\tilde{p}_{k,g+1},\tilde{p}_{kg})$, knowing that  $u_{ki}(x)\leq u_{ki}(\tilde{p}_{ki})$ for $x\in [\tilde{p}_{kg},\tilde{p}_{k,g-1})$ and $m_k-1\geq g\geq i+1$ (at the base we had $g=i+1$).
\begin{equation}
\begin{aligned}
\frac{1}{g+1}u_{k,g+1}(x)-\frac{1}{i}u_{ki}(x)&=A_{k,g+1,i}(x)\\
\frac{1}{g+1}u_{k,g+1}(\tilde{p}_{kg})-\frac{1}{i}u_{ki}(\tilde{p}_{kg})&=A_{k,g+1,i}(\tilde{p}_{kg})\\
u_{k,g+1}(x)&=u_{k,g+1}(\tilde{p}_{kg})
\end{aligned}
\nonumber
\end{equation}
Note that $\tilde{p}_{kg} = \tilde{v}_{k,g+1}$. Subtracting the first and the second equation, we get,

\footnotesize
\begin{equation}
\frac{1}{i}(u_{ki}(x)-u_{ki}(\tilde{p}_{kg}))=A_{k,g+1,i}(\tilde{p}_{kg})-A_{k,g+1,i}(x)\leq0
\nonumber
\end{equation}
\normalsize
Thus, $u_{ki}(x)\leq u_{ki}({\tilde{p}_{kg}})$ for $x \in [\tilde{p}_{k,g+1},\tilde{p}_{kg})$. The induction hypothesis yields  $u_{ki}(x)\leq u_{ki}(\tilde{p}_{ki})$ for $x\in [\tilde{p}_{k,g+1},\tilde{p}_{kg})$.

2) {$i\geq l_k+1$ and $x \in (\tilde{v}_{ki}, v]$}: We have just shown that $u_{ki}(x) \leq u_{ki}(\tilde{p}_{ki})$ for all $x \in [\tilde{p}, \tilde{p}_{ki})$. We now show the same for all $x \in (\tilde{v}_{ki}, v]$. The claim follows by vacuity for $i = l_k+1$, since $\tilde{v}_{ki} = v$. We therefore consider $i > l_k+1$. Since $\tilde{v}_{kj} = \tilde{p}_{k,j-1}$ for $ l_k+1 \leq j \leq m_k$, and $\tilde{v}_{k,l_k+1}=v$, any such $x$ is in $(\tilde{p}_{kg}, \tilde{p}_{k,g-1}]$ for some  $l_k+1< g < i$. We prove this claim by induction on $g$, starting with the base case of $g = i-1$. Let $x < v$.

\vspace{-7pt}
\begin{equation}
\begin{aligned}
\frac{1}{i}u_{ki}(x)-\frac{1}{i-1}u_{k,i-1}(x)&=A_{k,i,i-1}(x)\\
\frac{1}{i}u_{ki}(\tilde{p}_{k,i-1})-\frac{1}{i-1}u_{k,i-1}(\tilde{p}_{k,i-1})&=A_{k,i,i-1}(\tilde{p}_{k,i-1})\\
u_{k,i-1}(x)&=u_{k,i-1}(\tilde{p}_{k,i-1}) \end{aligned}
\nonumber
\end{equation}
Subtracting the first and the second equation, we get,

\vspace{-7pt}
\small
\begin{equation}
\frac{1}{i}(u_{ki}(x)-u_{ki}(\tilde{p}_{k,i-1}))=A_{k,i,i-1}(x)-A_{k,i,i-1}(\tilde{p}_{k,i-1})\leq 0
\nonumber
\end{equation}

\normalsize
Therefore $u_{ki}(x)\leq u_{ki}({\tilde{p}_{k,i-1}})$ for $x \in (\tilde{p}_{k,i-1}, \tilde{p}_{k,i-2}]\setminus {v}$. The claim is established in the base case if $\tilde{p}_{k,i-2} < v$. Else, if $\tilde{p}_{k,i-2} = v$, the claim has been shown only for $x \in (\tilde{p}_{k,i-1}, v)$ and we still need to show that $u_{ki}(v) \leq  u_{ki}(\tilde{p}_{k,i-1})$, which we proceed to do. Now, let $x = v$. if the seller $\bar{k}$ has a jump when it offers $l_{\bar{k}}+1$ units, since $i > l_k+1$,  for all sufficiently small but positive $\epsilon$, $u_{ki}(v) < u_{ki}(v-\epsilon)$,
and for sufficiently small but positive $\epsilon$, $v-\epsilon \in (\tilde{p}_{k,i-1}, v)$. Since $u_{ki}(v-\epsilon) \leq u_{ki}(\tilde{p}_{k,i-1})$, the base case follows. If not, that is seller $\bar{k}$ does not have a jump when it offers $l_{\bar{k}}+1$ units, using equation (\ref{equation:utility}) and continuity, we can deduce that $u_{ki}(v) \leq  u_{ki}(\tilde{p}_{k,i-1})$. The base case follows.

Now we want to prove that $u_{ki}(x)\leq u_{ki}(\tilde{p}_{k,i-1})$ for $x\in (\tilde{p}_{k,g-1},\tilde{p}_{k,g-2}]$, knowing that  $u_{ki}(x)\leq u_{ki}(\tilde{p}_{k,i-1})$ for $x\in (\tilde{p}_{kg},\tilde{p}_{k,g-1}]$ and $g\leq i-1$ and $g-1\geq l_k+1 $. First, let $x < v$.

\vspace{-7pt}
\small
\begin{equation}
\begin{aligned}
\frac{1}{i}u_{ki}(x)-\frac{1}{g-1}u_{k,g-1}(x)&=A_{k,i,g-1}(x)\\
\frac{1}{i}u_{ki}(\tilde{p}_{k,g-1})-\frac{1}{g-1}u_{k,g-1}(\tilde{p}_{k,g-1})&=A_{k,i,g-1}(\tilde{p}_{k,g-1})\\
u_{k,g-1}(x)&=u_{k,g-1}(\tilde{p}_{k,g-1})
\end{aligned}
\nonumber
\end{equation}

\normalsize
Subtracting the first and the second equation, we get,

\vspace{-7pt}
\small
\begin{equation}
\frac{1}{i}(u_{ki}(x)-u_{ki}(\tilde{p}_{k,g-1}))=A_{k,i,g-1}(x)-A_{k,i,g-1}(\tilde{p}_{k,g-1})\leq 0
\nonumber
\end{equation}

\normalsize
The inequality is because of the fact that $A_{k,l,j}(x)$ is non increasing and non positive with respect to $x$ if $l > j.$ Therefore $u_{ki}(x)\leq u_{ki}({\tilde{p}_{k,g-1}})$. Furthermore we know from the assumption of induction that $u_{ki}(\tilde{p}_{k,g-1})\leq u_{ki}(\tilde{p}_{k,i-1})$, thus  $u_{ki}(x)\leq u_{ki}(\tilde{p}_{k,i-1})$ for $x\in (\tilde{p}_{k,g-1},\tilde{p}_{k,g-2}]\setminus {v}$. We can show that $u_{ki}(v) \leq u_{ki}(\tilde{p}_{k,i-1})$ if $v \in (\tilde{p}_{k,g-1}, \tilde{p}_{k,g-2}]$ exactly as in the base case. The proof that for each $i \geq l_k+1$ each $x \in [\tilde{p}_{ki}, \tilde{v}_{ki})$ is a best response when a seller offers $i$ units is therefore complete. 

3) {$i\leq l_k$}: Now let $i \leq l_k$. Thus, $l_k > 0.$  Consider two cases:
\begin{itemize}
\item $l_k+l_{\bar{k}}= d-1$. Therefore $i\leq l_k= d-l_{\bar{k}}-1$. As we previously mentioned, utility $u_{ki}(.)$, is continuous not only in interval $[c,v)$, but also at price $v$, if $i\leq d-l_{\bar{k}}-1$. Using (\ref{equation:diff_decreasing}), and the fact that $A_{k,l,j}(x)$ is non increasing and non positive with respect to $x$, for $l > j$ and a similar argument to  case 1, we can get $u_{ki}(x) \leq u_{ki}(v)$ for all $x \in  [\tilde{p}, v)$. The result follows.

\item $l_k+l_{\bar{k}}= d$. Therefore $i\leq l_k= d-l_{\bar{k}}$. Since $l_k+l_{\bar{k}}+1>d$, neither $\Phi_{kl_k+1}(.)$ nor $\Phi_{\bar{k}l_{\bar{k}}+1}(.)$ have a jump at $v$, and $u_{ki}(.)$ is continuous in $[c,v]$. The result follows by a similar argument to that of in the previous case.
\end{itemize}

\end{proof}

\section{Computation of NE Strategies in an Asymmetric Setting}\label{section:duopolyalgorithm_asym}

In this section, we consider the general case in which the setting may not be symmetric. First we develop a framework to obtain the strategy profiles that satisfy the properties listed in Theorem~\ref{theorem:necessary} (Section~\ref{section:duopolyalgorithm_asym_framework}). Then, we compute these strategies for a simple case of an asymmetric market in which $m_1=m_2=d=3$ (Section~\ref{section:duopolyalgorithm_asym_example}). In Section~\ref{section:examples}, we show that the system may have multiple Nash equilibria.

\subsection{Framework for computation}\label{section:duopolyalgorithm_asym_framework}

In Theorem~\ref{theorem:necessity_random}, it has been proved that the properties listed in Theorem~\ref{theorem:necessary} are sufficient properties for a NE whether $d> \{m_1,m_1\}$ or $d=\max\{m_1,m_2\}$. In this section, we use Theorem~\ref{theorem:necessary} to obtain a framework to identify a set of Nash equilibria for the game. 

First, fix $l_1$ and $l_2$ (refer to Property \ref{property:difference}). In addition, note that Theorem~\ref{theorem:necessary} specifies the ordering of support sets for a seller and not the relative ordering of support sets of the two sellers. Thus, we fix an ordering of $\tilde{p}_{ki}$'s and $\tilde{p}_{\bar{k}j}$'s for $i\in\{l_k+1,\dots,m_k\}$ and $j\in\{l_{\bar{k}}+1,\dots,m_{\bar{k}}\}$ such that for seller $k$ and $\bar{k}$ the lower bounds are ordered with a decreasing relation with $i$ and $j$ respectively, and $\tilde{p}_{km_k}=\tilde{p}_{\bar{k}m_{\bar{k}}}=\tilde{p}$.
The unknowns that we should determine for a NE are  $\tilde{p}$, $m_k-l_k-1$ and $m_{\bar{k}}-l_{\bar{k}}-1$ number of lower bounds other than $\tilde{p}$ for seller $k$ and $\bar{k}$ respectively, and the distribution of price over each support set.

For these particular $l_1$, $l_2$, and relative ordering of support sets, based on Theorem~\ref{theorem:necessary}, the NE is the solution of:

\small
\begin{eqnarray}\label{equation:system}
\begin{aligned}
u_{ki}(\tilde{p}_{ki})&=u_{ki}(\tilde{p}^-_{k,i-1})\qquad \qquad  i\in A\\
u_{\bar{k}j}(\tilde{p}_{\bar{k}j})&=u_{\bar{k}j}(\tilde{p}^-_{\bar{k},j-1})\qquad \qquad  j\in A\\
u_{ki}(\tilde{p}_{ki})&=u_{ki}(\tilde{p}^-_{\bar{k}j})\qquad  i\in A,j:\tilde{p}_{{\bar{k}j}}\in (\tilde{p}_{ki},\tilde{p}_{k,i-1})\\
u_{\bar{k}j}(\tilde{p}_{\bar{k}j})&=u_{\bar{k}j}(\tilde{p}^-_{ki}) \qquad  j\in A,i:\tilde{p}_{ki}\in (\tilde{p}_{\bar{k}j},\tilde{p}_{\bar{k},j-1})\\
f_1f_2&=0
\end{aligned}
\end{eqnarray}

\normalsize
where $A=\{l_k+1,\dots,m_k\}$. In addition, $f_1$ and $f_2$ are the magnitude of jump at $v$ for the first and second seller when they offer $l_k+1$ and $l_{\bar{k}}+1$ units, respectively. Note that the first four sets of equations are derived using the fact that the utility of a seller should be equal over the entire support set. The fifth equation ensures that only one seller can have a positive jump at $v$.

In equation \eqref{equation:system}, the unknowns are $\tilde{p}$, $m_1+m_2-l_1-l_2-2$ number of lower-bounds other than $\tilde{p}$, $p_1$, $p_2$, and $m_1+m_2-l_1-l_2-2$ number of probability distributions at some specific points. That is $\Phi_{ki}(\tilde{p}_{\bar{k}j})$ for $i\in\{l_k+1,\dots,m_k\}$ and $j\text{ such that }\tilde{p}_{{\bar{k}j}}\in (\tilde{p}_{ki},\tilde{p}_{k,i-1})$. By solving the system of equations \eqref{equation:system}, we can get a candidate NE.

Using the solution, $\Phi_{ki}(.)$ for $k\in\{1,2\}$ and $i\in\{1,\dots, m_k\}$ can be found. To find the distributions of price for prices less than $v$, first note that each price $x\in[\tilde{p},v)$ which is not a lower bound for the support set belongs to exactly one of the support sets of each seller.  Therefore, by \eqref{equation:utility}, the expression of utility of player $k$ when it offers $i$ units depends only on $x$ and $\Phi_{\overline{k}j}(x)$, i.e. $u_{ki}(x)=(x-c)G(\Phi_{\overline{k}j}(x))$, where $G(\Phi_.(.))$ is a decreasing function of $\Phi_.(.)$, and therefore its inverse exists. On the other hand,  the utilities at the lower bounds are obtained from \eqref{equation:system} for both sellers. Using Property \ref{prop:equal}, $\Phi_{\bar{k}j}(x)=G^{-1}(\frac{u_{ki}(\tilde{p}_{kj})}{x-c})$. If the resulting $\Phi_{\bar{k}j}(\cdot)$ are valid probability  distribution functions, using Theorem \ref{theorem:necessity_random} we can conclude that they constitute a NE for the given $l_1$, $l_2$, and the fixed ordering of lower bounds.

We have shown how to obtain a Nash equilibrium given one exists for a particular choice of $l_1$, $l_2$, and a relative ordering between the support sets of the two sellers. 
Note that by changing the choices of the above we can possibly obtain multiple Nash equilibria. In the next sections, we present an example in which there exist at least two equilibria. 


\subsection{Example illustration of computation of Nash Equilibria} \label{section:duopolyalgorithm_asym_example}

Consider the case in which each seller offers up to three units and the total demand is exactly three units, i.e. $d=3$.  Without loss of generality we assume that $l_1 \geq l_2$; the strategy profiles in the other case $l_1 < l_2$ can be obtained by swapping the indices of the sellers.

1) First we focus on the case in which $l_1+l_2=d-1=2$. In this case,  $l_1=l_2=1$ or $l_1=2, l_2=0$.  If $l_1=l_2=1$, then sellers choose $v$ with probability $1$, if they offer $1$ unit of commodity. In order to specify the NE, we should find the lower bounds $\tilde{p}_{13}=\tilde{p}_{23}=\tilde{p}$, $\tilde{p}_{12}$, $\tilde{p}_{22}$, jumps at price $v$ ($f_1$ and $f_2$), and each distribution $\Phi_{kj}(.)$ for all $k=1,2$, and $j=2,3$.

First consider the ordering of lower bounds in which $\tilde{p}_{22}\geq \tilde{p}_{12}$ (Figure~\ref{figure:example10}). The system of equations is presented in the next page. Using equations (\ref{equ:1}), (\ref{equ:3}), (\ref{equ:5}), and (\ref{equ:6}), we can find $\tilde{p}_{22}$ as,

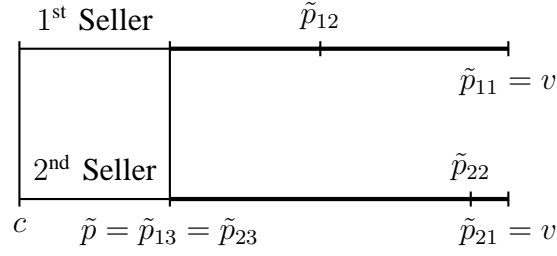
\begin{figure}
\begin{center}
\begin{tikzpicture}
\draw [thick] (0,2) -- (2,2);
\draw [ultra thick] (2,2) -- (6.5,2);
\draw [thick] (0,0) -- (2,0);
\draw [ultra thick] (2,0) -- (6.5,0);
\draw [thick] (0,-.1) node[below]{$c$} -- (0,2.1);
\draw [thick] (6.5,1.9) node[align=left,below]{$\tilde{p}_{11}=v$} -- (6.5,2.1);
\draw [thick] (6.5,-.1) node[align=left,below]{$\tilde{p}_{21}=v$} -- (6.5,.1);
\draw [thick] (2,-.1) node[below]{$\tilde{p}=\tilde{p}_{13}=\tilde{p}_{23}$} -- (2,2.1);
\draw [thick] (4,1.9) -- (4,2.1)  node[above]{$\tilde{p}_{12}$};
\draw [thick] (6,-.1)  -- (6,.1)node[above]{$\tilde{p}_{22}$};
\node[align=right, above] at (1,0.1)%
{$2^{\text{nd}}$ Seller};
\node[align=right, above] at (1,2.1)%
{$1^{\text{st}}$ Seller};
\end{tikzpicture}
\end{center}
\caption{Structure corresponding to $l_1=1$ and $l_2=1$}\label{figure:example10}
\end{figure}

\begin{figure*}[t]
\centering
\setlength\fboxsep{0pt}
\setlength\fboxrule{0.25pt}
\fbox{
 \addtolength{\linewidth}{-2\fboxsep}%
 \addtolength{\linewidth}{-2\fboxrule}%
 \begin{minipage}{\linewidth}
\scriptsize
\begin{equation} \label{equ:1}
u_{13}(\tilde{p})=u_{13}(\tilde{p}_{12})\Rightarrow 3(\tilde{p}-c)=(3-3q_{23}\Phi_{23}(\tilde{p}_{12}))(\tilde{p}_{12}-c)
\end{equation}
\begin{equation} \label{equ:2}
u_{23}(\tilde{p})=u_{23}(\tilde{p}_{12})\Rightarrow 3(\tilde{p}-c)=(3-3q_{13})(\tilde{p}_{12}-c)
\end{equation}
\begin{equation} \label{equ:3}
u_{23}(\tilde{p})=u_{23}(\tilde{p}_{22})\Rightarrow 3(\tilde{p}-c)=(3-3q_{13}-2q_{12}\Phi_{12}(\tilde{p}_{22}))(\tilde{p}_{22}-c)
\end{equation}
\begin{equation} \label{equ:4}
u_{12}(v^-)=u_{12}(\tilde{p}_{22})\Rightarrow (v-c)(2q_{20}+2q_{21}+2q_{22}f_2+q_{22}(1-f_2))=(\tilde{p}_{22}-c)(2-2q_{23})
\end{equation}
\begin{equation}\label{equ:5}
u_{12}(v^-)= u_{12}(\tilde{p}_{12})\Rightarrow (v-c)(2q_{20}+2q_{21}+2q_{22}f_2+q_{22}(1-f_2))=(\tilde{p}_{12}-c)(2-2q_{23}\Phi_{23}(\tilde{p}_{12}))
\end{equation}
\begin{equation} \label{equ:6}
u_{22}(v^-)=u_{22}(\tilde{p}_{22})\Rightarrow (v-c)(2q_{10}+2q_{11}+2q_{12}f_1+q_{12}(1-f_1))=(\tilde{p}_{22}-c)(2-2q_{13}-q_{12}\Phi_{12}(\tilde{p}_{22}))
\end{equation}
\begin{equation} \label{equ:7}
f_1f_2=0 \qquad \text{(At most one seller can have a jump at $v$ )}
\end{equation}
\end{minipage}
}
\caption*{System of equations for $l_1=l_2=1$ and $\tilde{p}_{22}\geq \tilde{p}_{12}$}
\end{figure*}

\normalsize

\vspace{-2mm}
\small
\begin{equation}\label{equ:p_22_1}
\begin{aligned}
\tilde{p}_{22}&=\frac{(v-c)A}{\frac{1}{2}-\frac{1}{2}q_{13}}+c\\
A&=\bigg{(}2q_{10}+2q_{11}+q_{12}(1+f_1)-\frac{3}{2}q_{20}-\frac{3}{2}q_{21}\\
&\qquad \qquad \qquad \qquad -\frac{3}{4}q_{22}(1+f_2)\bigg{)}
\end{aligned}
\end{equation}

\normalsize
On the other hand, from (\ref{equ:4}),\begin{equation}\label{equ:p_22_2}
\tilde{p}_{22}=\frac{(v-c)(2q_{20}+2q_{21}+q_{22}(1+f_2))}{2-2q_{23}}+c
\end{equation}

The values of $\tilde{p}_{22}$ in (\ref{equ:p_22_1}) and (\ref{equ:p_22_2}) should be equal. Utilizing this and (\ref{equ:7}),
\begin{equation}
\begin{aligned}
\frac{2f_1q_{12}}{1-q_{13}}-\frac{1}{2}q_{22}f_2A=(q_{20}&+q_{21}+\frac{1}{2}q_{22})A\\
&-\frac{4q_{10}+4q_{11}+2q_{12}}{1-q_{13}}=B
\end{aligned}
\end{equation}
where $A=\frac{1}{1-q_{23}}+\frac{3}{1-q_{13}}$. Therefore,

\begin{equation}
\left\{
	\begin{array}{ll}
		f_1=f_2=0  & \mbox{if } B=0 \\
		f_1>0 \& f_2=0 & \mbox{if } B>0 \\
		f_2>0 \& f_1=0 & \mbox{if } B<0 \\
	\end{array}
\right.
\end{equation}
Therefore $f_1$, $f_2$, and $\tilde{p}_{22}$ are uniquely determined. Using (\ref{equ:6}), $\Phi_{12}(\tilde{p}_{22})$ can be derived uniquely,

\vspace{-4mm}
\begin{equation}
\Phi_{12}(\tilde{p}_{22})=\frac{1}{q_{12}}\big{(}2-2q_{13}-\frac{v-c}{(\tilde{p}_{22}-c)}(2q_{10}+2q_{11}+q_{12}(1+f_1))\big{)}\\
\end{equation}

\normalsize
By (\ref{equ:3}), $\tilde{p}$ can be derived uniquely, (\ref{equ:2}) determines $\tilde{p}_{12}$ uniquely, and (\ref{equ:1}) provides us $\Phi_{23}(\tilde{p}_{12})$ uniquely. However, we should check whether $\Phi_{23}(\tilde{p}_{12})$ and $\Phi_{12}(\tilde{p}_{22})$ are between zero and one or not. If not, then this NE candidate is not valid. The distributions can be found by the process explained previously.

Another possible ordering of lower bounds is when $\tilde{p}_{22}\leq \tilde{p}_{21}$. The system of equations corresponding to this case can be obtained by swapping the index of sellers.

In the case of $l_1=2$ and $l_2=0$, Figure \ref{figure:example2} illustrates a schematic view of the support sets for the unique relative ordering of support sets. Equations can be obtained with a similar approach to the previous case.

\begin{figure}
\begin{center}
\begin{tikzpicture}
\draw [thick] (0,2) -- (2,2);
\draw [ultra thick] (2,2) -- (6,2);
\draw [thick] (0,0) -- (2,0);
\draw [ultra thick] (2,0) -- (6,0);
\draw [thick] (0,-.1) node[below]{$c$} -- (0,2.1);
\draw [thick] (6,1.9) node[align=left,below]{$\tilde{p}_{11}=\tilde{p}_{12}=v$} -- (6,2.1);
\draw [thick] (6,-.1) node[align=left,below]{$v$} -- (6,.1);
\draw [thick] (2,-.1) node[below]{$\tilde{p}=\tilde{p}_{13}=\tilde{p}_{23}$} -- (2,2.1);
\draw [thick] (4,-.1) -- (4,.1)  node[above]{$\tilde{p}_{22}$};
\draw [thick] (5,-.1) node[below]{$\tilde{p}_{21}$} -- (5,.1);
\node[align=right, above] at (1,0.1)%
{$2^{\text{nd}}$ Seller};
\node[align=right, above] at (1,2.1)%
{$1^{\text{st}}$ Seller};
\end{tikzpicture}
\end{center}
\caption{Structure corresponding to $l_1=2$ and $l_2=0$}\label{figure:example2}
\end{figure}


2) $l_1+l_2=3=d$. Note that $l_k=3$ and $l_{\bar{k}}=0$ can be ruled out since $l_k$ should be less than $m_k=3$. Thus, $l_1=2$ and $l_2=1$ (Figure~\ref{figure:example3}).  The approach to find the equilibria is similar to the previous cases.

\begin{figure}
\begin{center}
\begin{tikzpicture}
\draw [thick] (0,2) -- (2,2);
\draw [ultra thick] (2,2) -- (6,2);
\draw [thick] (0,0) -- (2,0);
\draw [ultra thick] (2,0) -- (6,0);
\draw [thick] (0,-.1) node[below]{$c$} -- (0,2.1);
\draw [thick] (6,1.9) node[align=left,below]{$\tilde{p}_{11}=\tilde{p}_{12}=v$} -- (6,2.1);
\draw [thick] (6,-.1) node[align=left,below]{$\tilde{p}_{21}=v$} -- (6,.1);
\draw [thick] (2,-.1) node[below]{$\tilde{p}=\tilde{p}_{13}=\tilde{p}_{23}$} -- (2,2.1);
\draw [thick] (5,-.1) node[below]{$\tilde{p}_{22}$} -- (5,.1);
\node[align=right, above] at (1,0.1)%
{$2^{\text{nd}}$ Seller};
\node[align=right, above] at (1,2.1)%
{$1^{\text{st}}$ Seller};
\end{tikzpicture}
\end{center}
\caption{Structure corresponding to $l_1=2$ and $l_2=1$}\label{figure:example3}
\end{figure}


\subsection{Multiple Nash Equilibria} \label{section:examples}
In Section~\ref{section:duopolyalgorithm}, we proved that the symmetric NE uniquely exists. In this section, we show that an asymmetric market allows for multiple Nash equilibria. Nash equilibria are computed using the above framework with  $v=10$ and $c=1$ and for different values of $\vec{q}_1$ and $\vec{q}_2$.  Some lead to a unique NE and some others to multiple Nash equilibria. For instance, the NE is unique, if
 \vspace{-2mm}
$$
\begin{aligned}
\vec{q}_1&=[0.45,0.1,0.4,0.05] \qquad
\vec{q}_2=[0.2,0.2,0.45,0.15]
\end{aligned}
$$
In this case, in the NE strategy,  $l_1=1$, $l_2=2$,  $\tilde{p}_{12}=9.0526$, $\tilde{p}=8.65$, and $\Phi_{23}(\tilde{p}_{12})=0.3333$, and the second seller has a jump of size $0.625$ at price $v=10$.
However, there are two Nash equilibria if:
\vspace{-2mm}
$$
\begin{aligned}
\vec{q}_1&=[0.05,0.1,0.4,0.45]\qquad \vec{q}_2&=[0.2,0.2,0.4,0.2]
\end{aligned}
$$
In both NE, $l_1=2$, $l_2=1$, and $\Phi_{13}(\tilde{p}_{22})=0.4444$. In the first NE, $f_2=0.06525$, $f_1=0$, $\tilde{p}=5.95$, and $\tilde{p}_{22}=7.1875$. In the second NE, $f_2=0$, $f_1=0.7778$, $\tilde{p}=5.8$, and $\tilde{p}_{22}=7$.


\section{Proof of Theorem~\ref{theorem:symmetricsuff}}

Before going to the proof of Theorem~\ref{theorem:symmetricsuff}, we need to prove some lemmas and theorems. First we prove that  $A_{l,j}(x)$ is (strictly) decreasing for $v>x\geq\tilde{p}_{m-1}$ when $d=m$ (Lemma~\ref{lemma:decreasing_random_asym2}). Then, in Lemma~\ref{lemma:minimumtildep}, we prove that the minimum of the lower end points is the lower end point of $\Phi_m(x)$, i.e., $\tilde{p}=\tilde{p}_m$. Next, using Lemmas~\ref{lemma:decreasing_random_asym2} and \ref{lemma:minimumtildep}, we prove that $\tilde{p}_i\notin [\tilde{p}_m,\tilde{p}_{m-1}) $ for $i\in\{1,\dots,m-2\}$.  This establishes the ordering for $\Phi_{m}(.)$ and $\Phi_{m-1}(.)$. After that we proceed to establish the ordering for the remaining support sets $\Phi_{j}(.)$ for $j\in\{1,\dots,m-2\}$, knowing that for them $\tilde{p}_j\geq \tilde{p}_{m-1}$. A similar result to the Property~\ref{property:supportsets} is proved in Property~\ref{property:supportsets2}. Finally, we prove Theorem~\ref{theorem:symmetricsuff}.

Note that a symmetric NE in a symmetric market is considered in this section. Let us define $A_{l,j}(x)=\frac{1}{l}u_{l}(x)-\frac{1}{j}u_{j}(x)$. $B_{l,j}(x)$ is defined such that,
$$A_{l,j}(x)=(x-c)B_{l,j}(x)$$
where,

\vspace{-10pt}
\small
\begin{equation}\label{equation:diff_decreasing2}
\begin{aligned}
B_{l,j}(x)=-\frac{1}{l}&\sum_{i=d-l+1}^{d-j}{\Phi_{i}(x)q_{i}(i-d+l)} +\\&\qquad \sum_{i=d-j+1}^{m}{\Phi_{i}(x)q_{i}(d-i)(\frac{1}{l}-\frac{1}{j})}
\end{aligned}
\end{equation}
\normalsize
 Based on the following lemma, $A_{l,j}(x)$ is (strictly) decreasing for $v>x\geq\tilde{p}_{m-1}$ and $l>j$, when $d=m$.

\begin{lemma}\label{lemma:decreasing_random_asym2}
For every $l$ and $j$, $l>j\geq 1$, $A_{l,j}(x)$ is (strictly) decreasing for $v>x\geq\tilde{p}_{m-1}$ when $d=m$.
\end{lemma}

We argued that $B_{l, j}(\cdot)$ is non increasing and non positive with respect to the price $x$. To prove that  $A_{l,j}(.)=(x-c)B_{l,j}(x)$ is strictly decreasing, it is enough to prove that  $B_{l, j}(\cdot)$ is negative. We will prove that $\Phi_{m-1}(x)$ is included in the summation of $B_{l, j}(\cdot)$ and obviously positive for $x> \tilde{p}_{m-1}$. In addition, its coefficient is negative since $d=m>m-1$. Thus, the result follows.

\begin{proof}
It is enough to prove that $B_{l,j}(x)$ is non-increasing for $x\geq \tilde{p}_{m-1}$ and negative for $x>\tilde{p}_{m-1}$ when demand is $m$. This yields that  $A_{l,j}(x)=(x-c)B_{l,j}(x)$ is strictly decreasing with respect to $x$. 

Note that in (\ref{equation:diff_decreasing2}), $\Phi_{i}(.)$'s are non-negative and non-increasing since they are probability distributions. In addition, they have non-positive weights: $-(i-d-l)\leq -1<0$, $\frac{1}{l}-\frac{1}{j}<0$, and $d-i\geq d-m=0$ (note that $d=m$). Thus $B_{l,j}(x)$ is non increasing and non positive with respect to the price $x$ when $l\geq j$. To prove that $B_{l,j}(x)$ is negative for $x>\tilde{p}_{m-1}$, since $d-(m-1)=1>0$ and $-(i-d-l)\leq -1<0$ (possible coefficients of $\Phi_{m-1}(x)$) , it is enough to prove that $\Phi_{m-1}(.)$ is included in the summation of $B_{l,j}(.)$  and it is positive, i.e. $\Phi_{m-1}(x)>0$ for  $x> \tilde{p}_{m-1}$. The later follows from the definition of $\tilde{p}_{m-1}$.

Now we prove that $\Phi_{m-1}(.)$ is included in the summation of $B_{l,j}(.)$. Note that $l>j\geq 1$. Thus $l\geq 2$, and the lowest index of the \eqref{equation:diff_decreasing2} is $d-l+1\leq m-2+1=m-1$. The result follows.
\end{proof}

To prove the ordering and disjoint properties in the symmetric setting we should alter the proofs. First we will prove that $\tilde{p}=\tilde{p}_m$, i.e. the minimum of lower bounds is the lower bound of $\Phi_{m}(x)$. Then we will prove that $\tilde{p}_j\notin [\tilde{p}_m,\tilde{p}_{m-1}) $ for $j\in\{1,\dots,m-2\}$. This proves that the next lowest support set is the support set of $\Phi_{m-1}(.)$. After that using Lemma \ref{lemma:decreasing_random_asym} will prove that the support set of $\Phi_l(.)$ for $l<m$ is a subset of $[\tilde{p}_{m-1},p_j]$ for all integers $j\in[1,l)$.   These three all together establishes the ordering.

\begin{lemma}\label{lemma:minimumtildep}
$\tilde{p}=\tilde{p}_m$, i.e. the minimum of lower end points is the lower end point of $\Phi_{m}(x)$.
\end{lemma}

\begin{proof}
Suppose not and there exists $x>\tilde{p}$ such that $x\leq \tilde{p}_{m}$. By Property \ref{property:contiguous_assym}, there exists an $\epsilon>0$ and an availability level $j\neq m$ such that $[\tilde{p}_{m}-\epsilon,\tilde{p}_{m}]$ belongs to the support set of $\Phi_{j}(.)$ and $\tilde{p}_{j}<\tilde{p}_{m}$.  Thus $u_{j}(\tilde{p}_{m})=u_{j}(\tilde{p}_{m}-\epsilon)$. In addition, $B_{m,j}(x)$ is the weighted summation of $\Phi_{i}(.)$ for $i\in\{1,\dots,m\}$. Thus, the distribution $\Phi_{j}(.)$ is included in the summation of $B_{m,j}(x)$, and its coefficient is negative. In addition, $\Phi_j(x)>0$ for $x>\tilde{p}_j$.   
 Thus, $A_{m,j}(x)$ is strictly decreasing with respect to $x$ for $x>\tilde{p}_{j}$. Thus $A_{m,j}(\tilde{p}_{m}-\epsilon)>A_{m,j}(\tilde{p}_{m})$. Note that $u_{j}(\tilde{p}_{m})=u_{j}(\tilde{p}_{m}-\epsilon)$. Thus, $u_{m}(\tilde{p}_{m})=u_{m,max}< u_{m}(\tilde{p}_{m}-\epsilon)$. This contradicts with $\tilde{p}_{m}$ belonging to the support set of $\Phi_{m}(.)$. The result follows.


\end{proof}

\begin{lemma}\label{lemma:oderm}
$\tilde{p}_i\notin [\tilde{p}_m,\tilde{p}_{m-1}) $ for $i\in\{1,\dots,m-2\}$.
\end{lemma}
 
To prove this, we use a contradiction argument. Suppose that there exists $\tilde{p}_j\in [\tilde{p}_m,\tilde{p}_{m-1})$ such that $j\in\{1,\dots,m-2\}$. We will prove that no $x\in(\tilde{p}_j,\tilde{p}_{m-1}]$ is in the support of $\Phi_{m}(.)$. Thus there exists  $u\in\{1,\dots,m-2\}$ such that $\tilde{p}_{m-1}$ is in the support set of $\Phi_u(.)$. We prove that the payoff of the seller when she offers $u$ units with price  $\tilde{p}_{m-1}+\epsilon$ is strictly greater than the payoff when offering  with price $\tilde{p}_{m-1}$.
 This is in contradiction with $\tilde{p}_{m-1}$  being the best response for player with availability $u$.

\begin{proof}
The lemma follows by vacuity if $m\leq 2$. Take $m>2$. Note that $\tilde{p}_{m-1}<v$. If not there is a jump of size $1$ at price $v$ when the seller offers $m-1$ units. Since $2m-2>d=m$ for $m>2$, using Lemma \ref{lemma:jump_asy}, $u_{m-1}(v-\epsilon)>u_{m-1}(v)$ for $\epsilon$ small enough. This is in contradiction with assigning a positive probability to price $v$ in the equilibrium when seller offers $m-1$ units. Thus $\tilde{p}_{m-1}< v$. 

Suppose there exists $\tilde{p}_j\in [\tilde{p}_m,\tilde{p}_{m-1})$ such that $j\in\{1,\dots,m-2\}$. We will prove that no $x\in(\tilde{p}_j,\tilde{p}_{m-1}]$ is in the support of $\Phi_{m}(.)$. Thus (using this and Property \ref{property:contiguous_assym}), there exists  $u\in\{1,\dots,m-2\}$ such that $\tilde{p}_{m-1}$ is in the support set of $\Phi_u(.)$. Consider $B_{m-1,u}(x)$ which is the summation of weighted distributions $\Phi_{i}(x)$ when $i\in \{2,\dots,m-1\}$. Thus, the distribution $\Phi_{m-1}(.)$ is included in the summation of $B_{m-1,u}(x)$ (note that $m>2$), and its coefficient is negative (Note that $d>0$).
 Thus, $A_{m-1,u}(x)$ is strictly decreasing with respect to $x$ for $x>\tilde{p}_{m-1}$. Thus $A_{m-1,u}(\tilde{p}_{m-1}+\epsilon)<A_{m-1,u}(\tilde{p}_{m-1})$. Using $u_{m-1}(\tilde{p}_{m-1})=u_{m-1}(\tilde{p}_{m-1}+\epsilon)$, we can conclude that $u_{u}(\tilde{p}_{m-1})=u_{u,max}< u_{u}(\tilde{p}_{m-1}+\epsilon)$. This is in contradiction with $\tilde{p}_{m-1}$  being the best response for player with availability $u$. Note that $\tilde{p}_{m-1}<v$, and every price less than $v$ which belongs to the support set of a distribution $\Phi_i(.)$ should be a best response for players when offering $i$ units. The lemma follows.

Now we complete the proof by proving that no $x\in(\tilde{p}_j,\tilde{p}_{m-1}]$ is in the support of $\Phi_{m}(.)$. Suppose not. We will show that  
there exist an availability level $f$ and two prices $y_1$ and $y_2$, such that $\tilde{p}_j<y_1<\tilde{p}_{m-1}$, belongs to the support set of $\Phi_m(.)$, and both $y_1$ and $y_2$ belong to the support set of $\Phi_f(.)$. Then we will show that $u_m(y_1)<u_m(y_2)$, which contradicts with $y_1$ being in the support set of $\Phi_m(.)$. 

Using the contradiction assumption, $w$ is defined as,
\vspace{-2mm}
$$
w=\inf_{x\in (\tilde{p}_j,\tilde{p}_{m-1}] \text { \& $x\in$ Supp($\Phi_m(.)$)}}{x}
$$
Note that $w$ is in the support set of $\Phi_m(.)$. Now consider two cases:

\begin{enumerate}
\item $w>\tilde{p}_j$: Using continuity, the definition of support sets, and Property~\ref{property:contiguous_assym}, there exist $\epsilon$ and $f\in\{1,\dots,m-2\}$ such that $w$ and $w-\epsilon$ is in the support set of  $\Phi_f(.)$. Take $y_1=w$ and $y_2=w-\epsilon$. 

\item $w=\tilde{p}_j$: Using continuity and the definition of infimum, there exists $\epsilon$ such that every $w+\epsilon$ belong to the support set of  $\Phi_m(.)$ and $\Phi_j(.)$. Take $f=j$, $y_1=w+\epsilon$, and $y_2=w$.
\end{enumerate}

Next, we will prove that $u_m(y_1)<u_m(y_2)$, which contradicts with $y_1$ being in the support set of $\Phi_m(.)$. Note that $y_1<v$, and every price less than $v$ which belongs to the support set of a distribution $\Phi_i(.)$ should be a best response for players when offering $i$ units. This completes the proof.

Consider $B_{m,f}(x)$ which is the summation of weighted distributions $\Phi_{i}(x)$ when $i\in \{1,\dots,m-1\}$. Thus, the distribution $\Phi_{f}(.)$ is included in the summation of $B_{m,f}(x)$, and its coefficient is negative. Thus, $A_{m,f}(x)$ is strictly decreasing with respect to $x$ for $x\geq\tilde{p}_{f}$. Thus $A_{m,f}(y_2)>A_{m,f}(y_1)$. Using $u_{f}(y_1)=u_{f}(y_2)$, we can conclude that $u_{m}(y_1)< u_{m}(y_2)$. The contradiction argument is complete.

\end{proof}

Therefore we established the ordering for $\Phi_{m}(.)$ and $\Phi_{m-1}(.)$. Now we are set to establish the ordering for the remaining support sets $\Phi_{j}(.)$ for $j\in\{1,\dots,m-2\}$, knowing that for them $\tilde{p}_j\geq \tilde{p}_{m-1}$. The next is the counterpart of the Property~\ref{property:supportsets2} in symmetric setting.
\begin{property}\label{property:supportsets2}
The support set of $\Phi_{l}(.)$ is a subset of $[\tilde{p},\tilde{p}_{j}]\cup [v]$ for all integers $j \in [1,l)$.
\end{property}

\begin{proof}
Consider support sets of $\Phi_{j}(\cdot)$, $\Phi_{l}(\cdot)$, and $j<l$. We will show that $u_{l}(a)< u_{l}(\tilde{p}_{j})$ for all $a\in (\tilde{p}_{j}, v)$. Thus, no $a\in (\tilde{p}_{j},v)$ is a best response for the seller with availability of $l$ units. Therefore, the support set of $\Phi_{l}(\cdot)$ is a subset of $[\tilde{p},\tilde{p}_{j}]\cup [v]$.

We now complete the proof, by showing that $u_{l}(a) < u_{l}(\tilde{p}_{j})$ for all $a \in(\tilde{p}_{j}, v)$:
\vspace{-2mm}
\begin{equation}
\begin{aligned}
\frac{1}{l}u_{l}(a)-\frac{1}{j}u_{j}(a)&=A_{l,j}(a)
\end{aligned}
\nonumber
\end{equation}

Note that if $\tilde{p}_j\geq v$, property follows by vacuity. Now we consider $\tilde{p}_j<v$. Since $j<l\leq m$, $j\leq m-1$. By Lemma~\ref{lemma:oderm},  $\tilde{p}_{m-1}\leq \tilde{p}_{j}<a<v$, by Lemma \ref{lemma:decreasing_random_asym2}, $A_{l,j}(a)$ is decreasing function of $a$ for $a \in [\tilde{p}_{m-1}, v)$. Thus, $A_{l,j}(a) < A_{l,j}(\tilde{p}_{j})$ for $a\in(\tilde{p}_{j}, v)$. On the other hand $u_{j}(a)\leq u_{j}(\tilde{p_{j}})$ for all $a>\tilde{p}_{j}$ , since $\tilde{p}_{j}$ is a best response of a seller with availability $j$, therefore $u_{l}(\tilde{p}_{j})>u_{l}(a)$.
\end{proof}

Now we will prove the Theorem~\ref{theorem:symmetricsuff}:


\begin{proof}
Note that the first place that we used the condition  $d>\max\{m_1,m_2\}$ (in symmetric setting $d>m$) instead of  $d=\max\{m_1,m_2\}$ ($d=m$) was in Section~\ref{subsection:dgeqm}. Thus all of the results before that apply also to the case that  $d=m$. Property~\ref{property:supportsets2} provides exactly the same property in the Property~\ref{property:supportsets} for the symmetric scenario. Thus the corollaries after the property follows. In addition, results in the Section~\ref{subsection:difference} follows, since they are based on results before the Section~\ref{subsection:dgeqm} and Property~\ref{property:supportsets} and its corollaries. Thus Theorem~\ref{theorem:necessary} goes through in the case of a symmetric NE and $d=m$. 
\end{proof}

\section{Oligopoly Market}\label{appendix:oligopoly}

Suppose that the setting is symmetric and there exist $n$ sellers in the market. We consider a strategy that satisfies the properties identified for a symmetric NE in Section~\ref{section:duopolyalgorithm} with the difference that in our proposed strategies the threshold $l^*=\lfloor \frac{d}{n}\rfloor$. Note that the algorithm for finding such a strategy is similar to what is presented in Section~\ref{section:symm_alg}, but the results would be different. We now investigate how well this strategy approximates an NE strategy in an oligopoly market. 

We numerically compute the maximum expected utility for a particular seller, when all other sellers choose the proposed strategy (best response utility, $U_{\text{Best Response}}$). We observe that over a large set of parameters for all possible availability levels, the best response utility is either the same as the expected utility obtained by following the proposed strategy ($U_{\text{Proposed Strategy}}$), or is fairly close to this value \footnote{For large sets of parameters, the difference is at most $5$ percent of the value of the expected utility resulted by the proposed strategy.}.

For instance, consider a market in which the availability of each seller follows a binomial distribution, $\mathcal{B}(m,p)$, with binomial probability $p=0.4$ and $m=3$ ($m$ is the maximum possible available units with each seller). In addition, in this market the demand is $d=\max\{n,m\}$, $v=10$, and $c=1$. We plot the relative difference, described as follows, between the best response utility and the expected utility of the proposed strategy versus different number of sellers, i.e. $n$, for different availability levels in Figure~\ref{figure:best_response}.
\vspace{-2mm}
$$
\text{Relative Difference}=\frac{U_{\text{Best Response}}-U_{\text{Proposed Strategy}}}{U_{\text{Proposed Strategy}}}
$$


\begin{figure}[t]
\begin{center}
\includegraphics[scale=0.36]{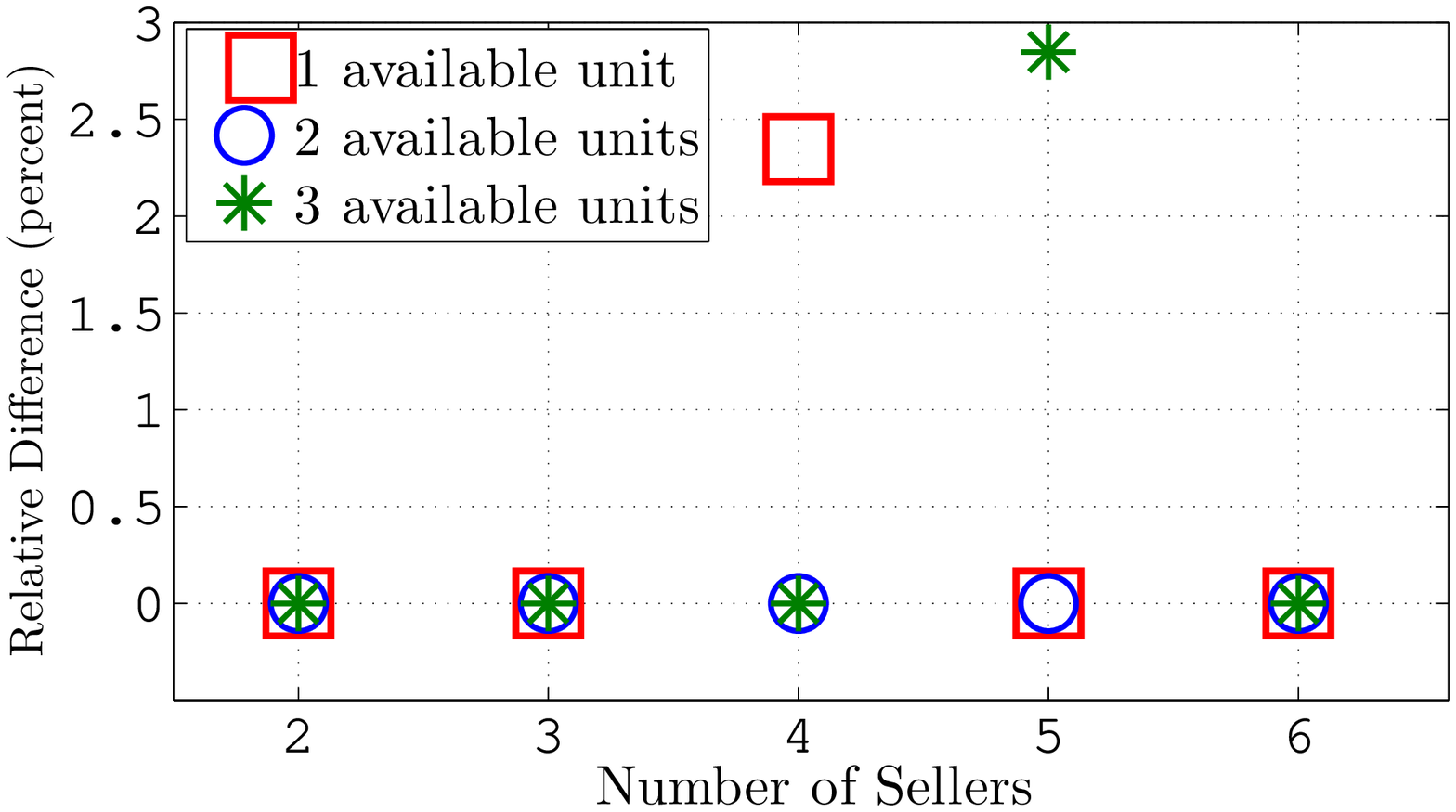}
\end{center}
\caption{The relative difference of the best response expected utility and the expected utility of the proposed strategy versus different number of sellers}
\label{figure:best_response}
\end{figure}

Note that the relative difference is zero for all availability levels when there exist $2$, $3$, and $6$ sellers in the market. Thus, the proposed strategy is a NE of the market in these cases. Although, in the case of 4 and 5 sellers the proposed strategy is not an NE when a seller has 1 and 2 units of commodity available, respectively, the relative difference in these cases is less than 3 percent. Thus, overall, we can say that the proposed strategy is a good approximations of NE when sellers are not too concerned about optimizing over small gains.

\bibliographystyle{IEEEtran}
\bibliography{bmc_article}

\end{document}